\documentclass[11pt]{amsart}
\usepackage{amssymb}

\newcommand{\affiliation}[1]{\address{#1}}

\usepackage[pdftex]{graphicx}

\DeclareMathOperator{\Tr}{Tr}

\newcommand\tD{\tau_{\mathrm{d}}}

\newcommand{\rme}{\mathrm{e}}
\newcommand{\rmi}{\mathrm{i}}
\newcommand{\T}{\mathrm{T}}
\newcommand{\U}{\mathrm{U}}
\newcommand{\OO}{\mathrm{O}}
\newcommand{\St}{Z}
\newcommand{\Tt}{T}

\newcommand{\widebar}[1]{\overline{#1}}
\newcommand{\wb}[1]{{\overline{#1}}}
\newcommand{\wt}[1]{{\widetilde{#1}}}
\newcommand{\bs}[1]{{\boldsymbol{#1}}}

\newcommand{\V}{v}
\newcommand{\E}{e}

\newcommand{\id}{{id}}

\newtheorem{theorem}{Theorem}[section]
\newtheorem{lemma}[theorem]{Lemma}

\theoremstyle{definition}
\newtheorem{definition}[theorem]{Definition}
\newtheorem{remark}[theorem]{Remark}
\newtheorem{example}[theorem]{Example}
\newtheorem{notation}[theorem]{Notation}

\begin{document}

\title[Combinatorial theory of the semiclassical evaluation of transport moments I]
{Combinatorial theory of the semiclassical evaluation of transport moments I: \\ Equivalence with the random matrix approach}

\author{G.~Berkolaiko} 
\affiliation{Department of Mathematics, Texas A\&M University, College
  Station, TX 77843-3368, USA}
\email{berko@math.tamu.edu}

\author{J.~Kuipers}
\affiliation{Institut f\"ur Theoretische Physik, Universit\"at Regensburg, D-93040
Regensburg, Germany}
\email{Jack.Kuipers@physik.uni-regensburg.de}

\begin{abstract}
  To study electronic transport through chaotic quantum dots, there
  are two main theoretical approachs.  One involves substituting the
  quantum system with a random scattering matrix and performing
  appropriate ensemble averaging. The other treats the transport in
  the semiclassical approximation and studies correlations among
  sets of classical trajectories.  There are established evaluation
  procedures within the semiclassical evaluation that, for several
  linear and non-linear transport moments to which they were applied,
  have always resulted in the agreement with random matrix
  predictions.  We prove that this agreement is universal: any
  semiclassical evaluation within the accepted procedures is
  equivalent to the evaluation within random matrix theory.  

  The equivalence is shown by developing a combinatorial
  interpretation of the trajectory sets as ribbon graphs (maps) with
  certain properties and exhibiting systematic cancellations among
  their contributions.  Remaining trajectory sets can be identified
  with primitive (palindromic) factorisations whose number gives the
  coefficients in the corresponding expansion of the moments of random
  matrices.  The equivalence is proved for systems with and
  without time reversal symmetry.

  %%CSE
  %The equivalence can be proved for all three
  %classical symmetry classes corresponding to systems with and without
  %time reversal and spin-rotation symmetry.
\end{abstract}

\maketitle

%%%%%%%%%%%%%%%%%%%%%%
\section{Introduction}

Transport through a chaotic cavity is usually studied through a
scattering description.  For a chaotic cavity attached to two leads
with $N_1$ and $N_2$ channels respectively, the scattering matrix is
an $N\times N$ unitary matrix, where $N=N_1+N_2$.  It can be separated
into transmission and reflection subblocks
\begin{equation}
\label{scatmatpartseqn}
S(E) = \left(\begin{array}{cc}\bs{r}&\bs{t}' \\ \bs{t} & \bs{r}'\end{array}\right),
\end{equation}
which encode the dynamics of the system and the relation between the
incoming and outgoing wavefunctions in the leads.
The transport statistics of the cavity in question can now be
expressed in terms of the subblocks of $S(E)$.  For example, the 
conductance is proportional to the trace $\Tr \left[\bs{t}^\dagger \bs{t}\right]$ 
(Landauer-B\"uttiker formula \cite{Buttiker86,Landauer57,Landauer88}), while other physical
properties are expressible through higher moments like $\Tr \left[\bs{t}^\dagger \bs{t}\right]^n$.

There are two main approaches to studying the transport statistics in clean ballistic systems:
a random matrix theory (RMT) approach, which argues that $S$ can be viewed as a
random matrix from a suitable ensemble, and a semiclassical approach
that approximates elements of the matrix $S$ by sums over open scattering
trajectories through the cavity.

It was shown by Bl\"umel and Smilansky
\cite{BluSmi_prl88,BluSmi_prl90} that the scattering matrix of a
chaotic cavity is well modelled by the Dyson's circular ensemble of
random matrices of suitable symmetry.  Thus, transport properties of
chaotic cavities are often treated by replacing the scattering matrix
with a random one (see \cite{beenakker97} for a review).  Calculating
the averages over the appropriate random matrix ensemble is a very
active area with many different approaches.  A partial list of recent results
include the papers \cite{kss09,lv11,ms11,ms13,novaes08,ok08,ok09,ss06,ssw08,vv08}.  
In Section~\ref{sec:RMT} we review some basic facts about integration 
over random matrices.

On the other hand, the semiclassical approach makes use of the following
approximation for the scattering matrix elements
\cite{miller75,richter00,rs02}
\begin{equation} 
  \label{scatmatsemieqn}
  S_{oi}(E) \approx \frac{1}{\sqrt{N\tD}}\sum_{\gamma (i \to o)}
  A_{\gamma}(E)\rme^{\frac{\rmi}{\hbar}S_{\gamma}(E)} ,
\end{equation}
which involves the open trajectories $\gamma$ which start in channel
$i$ (for ``input'') and end in channel $o$ (for ``output''), with
their action $S_{\gamma}$ and stability amplitude $A_{\gamma}$.  The
prefactor also involves $\tD$ which is the average dwell time, or time
trajectories spend inside the cavity.  For transport moments one
considers quantities of the type
\begin{equation}
  \label{semitrajeqn}
  M_n(X) = \left\langle\Tr \left[X^{\dagger}X\right]^{n} \right\rangle_E 
  \sim \left\langle \frac{1}{{(N\tD)}^{n}}
    \sum_{{i_j,o_j}} 
    \sum_{\substack{\gamma_j(i_j\to o_j) \cr
        \gamma'_j (i_{j+1}\to o_{j})}} 
    \prod_{j=1}^{n} A_{\gamma_j}A_{\gamma'_j}^{*}
    \rme^{\frac{\rmi}{\hbar}(S_{\gamma_j}-S_{\gamma'_j})} \right\rangle,
\end{equation}
where the trace means we identify $i_{n+1}=i_1$ and where $X$ is
either the transmitting or the reflecting subblock of the scattering
matrix.  The averaging is performed over a window of energies $E$
which is classically small but semiclassically large:
the width $\Delta E$ of the window satisfies $\hbar/N\tD \ll \Delta E \ll E$.
Note that we use dagger to mean conjugate-transpose of a matrix and
star to denote complex conjugation.

The choice of the subblock $X$ affects the sums over the possible
incoming and outgoing channels, but not the trajectory structure which
involves $2n$ classical trajectories connecting channels.  Of these,
$n$ trajectories $\gamma_j$, $j=1,\ldots,n$, contribute with positive
action while $n$ trajectories $\gamma'_j$ contribute with negative
action.  In the semiclassical limit of $\hbar\to0$ we require that
these sums approximately cancel on the scale of $\hbar$ so that the corresponding
trajectories can contribute consistently when we apply the averaging
in \eqref{semitrajeqn}.

The semiclassical treatment will be reviewed in
Section~\ref{sec:semiclassics}.  The main idea of the treatment is
that, in order to achieve a small action difference, the trajectories
$\{\gamma'_j\}$, must follow the path of trajectories $\{\gamma_j\}$
most of the time, deviating only in small regions called
\emph{encounters}.  The topological configuration of encounters and
trajectories' stretches between them is described using a
\emph{diagram}.  The task of semiclassical evaluation thus splits into
two parts: evaluation of the contribution of a given diagram by
integrating over all possible trajectories of given structure and
enumerating all possible diagrams.

For the former task, a well established approximation emerged by
extending the pioneering work of Richter and Sieber \cite{rs02} by a
group of physicists based mainly in Essen: S.\ M\"uller, S.\ Heusler,
P.\ Braun and F.\ Haake, \cite{heusleretal06,mulleretal07}.  Henceforth
we refer to this approximation as the ``Essen ansatz''.  Roughly speaking,
it assigns to each diagram a weight which depends on the number and
type of encounters and the number of trajectory stretches between the
encounters.  The approximation is derived based on physically
justified assumptions; a mathematical derivation remains outside reach
even for the simplest chaotic systems.

This approach to quantum transport was founded on the earlier 
semiclassical treatment of the two-point correlator of spectral 
densities of closed systems.  There, correlations between pairs of 
periodic orbits \cite{berry85,ha84,mulleretal04,mulleretal05,sr01} or 
sets of pseudo-orbits \cite{heusleretal07,mulleretal09} could be shown 
to be responsible for the universal behaviour in line with RMT.  In fact 
a one-to-one mapping exists between the semiclassical diagrams and those 
that arise in a perturbative expansion of the $\sigma$ model and hence 
the equivalence between semiclassics and RMT established \cite{mulleretal09,mulleretal05}.
The full equivalence for arbitrary correlators is yet to be demonstrated however.

Returning to transport through open systems, within the Essen ansatz 
incremental progress has been made in evaluating various transport-related 
quantities, see for example
\cite{bhn08,bk10,bk11,braunetal06,heusleretal06,kuipersetal11,kuipersetal10,mulleretal07,rs02}.  
Interestingly, transport diagrams can be related to the closed periodic 
orbit pairs of spectral statistics by connecting the outgoing and incoming 
channels (possibly with additional encounters introduced).  Reversing the 
process, transport diagrams may be generated by cutting periodic orbits and this 
approach has been employed to derive the first two moments 
\cite{braunetal06,heusleretal06,mulleretal07}.  For higher moments 
this procedure becomes more involved, though it does lead to interesting 
combinatorial problems \cite{novaes12,novaes13}.  Alternatively, transport 
diagrams can be generated without recourse to periodic orbits 
leading to a perturbative expansion of moment generating functions for 
several orders in the parameter $1/N$ \cite{bhn08,bk11,bk13b}.  Here instead 
we develop a combinatorial approach to directly describe all transport diagrams.

For the transport moments that have been evaluated semiclassically, 
in every case where a RMT prediction was available it was found to be in
full agreement with the semiclassical evaluation.  This paper is devoted to 
proving a general theorem that implies that this will always remain the case: 
\emph{any} semiclassical evaluation within the Essen ansatz 
is \emph{equivalent} to the RMT evaluation.

Before we formulate our theorem, we note that the trace of any form
can be expanded as a sum of products of matrix elements.  For example,
the trace in equation \eqref{semitrajeqn} expands as
\begin{align}
  \label{eq:trace_expansion}
  \Tr \left[X^{\dagger}X\right]^{n}
  &= \sum_{i_j, o_j}  S^\dagger_{i_1,o_n} S_{o_n,i_n} \cdots
  S_{o_2,i_2} S^\dagger_{i_2,o_1}S_{o_1,i_1} \\
  &= \sum_{i_j, o_j}  S_{o_n,i_n} \cdots S_{o_2,i_2} S_{o_1,i_1}
  S^*_{o_n,i_1} \cdots
  S^*_{o_2,i_3} S^*_{o_1,i_2} \nonumber \\
  &= \sum_{i_j, o_j}  \St_{i_1,o_1} \St_{i_2,o_2} \cdots \St_{i_n,o_n}
  \St^*_{i_2,o_1} \St^*_{i_3,o_2}\cdots
  \St^*_{i_1,o_n} \nonumber
\end{align}
where $X$ is a sub-block of the scattering matrix $S$, the variables
$i_1,\ldots, i_n$ run over all columns of $S$ that appear in $X$ and
the variables $o_1,\ldots, o_n$ run over all rows of $S$ in $X$.  In
the last line we switched from the matrix $S$ to its transpose $\St=S^{\T}$ as
this will make further notation less confusing.  From
equation \eqref{eq:trace_expansion} it is clear that if we can
evaluate the averages of products of matrix elements of $S$ (or $\St$),
we can have access to every linear transport moment.  The same goes
for the nonlinear transport moments, i.e.\ averages of the form
$\langle \Tr \left[X^{\dagger}X\right]^{n_1} \Tr
\left[X^{\dagger}X\right]^{n_2} \cdots \rangle$.

\begin{theorem}
  \label{thm:main_equivalence}
  Within the Essen ansatz, the energy average of a product of the elements
  of the scattering matrix $S(E)$ (or $\St(E) = S^{\T}(E)$) coincides with
  the corresponding average in RMT,
  \begin{equation}
    \label{eq:main_equality}
    C_E(\boldsymbol{a},\boldsymbol{b}) \equiv
    \langle \St_{a_1a_\wb{1}}\ldots \St_{a_sa_\wb{s}} 
    \St^*_{b_1b_\wb{1}}\ldots \St^*_{b_{t}b_\wb{t}}
    \rangle_{E}
    = \langle U_{a_1a_\wb{1}}\ldots U_{a_sa_\wb{s}} 
    U^*_{b_1b_\wb{1}}\ldots U^*_{b_{t}b_\wb{t}}
    \rangle_{\mathrm{RMT}}.
 \end{equation}
 For scattering matrices with time-reversal symmetry the appropriate
 RMT average is over Circular Orthogonal Ensemble; without the
 symmetry the RMT average is over Circular Unitary Ensemble.
\end{theorem}

We emphasize that equality~\eqref{eq:main_equality} is established to
all orders in $1/N$, the usual expansion parameter of the quantities
involved.

To prove Theorem~\ref{thm:main_equivalence} we will put the
diagrammatic method of the semiclassical approximation on a more
rigorous footing, carefully describing diagrams as ribbon graphs with
certain colorability properties.  We define operations on the diagrams
that lead to cancellations and, eventually, to the equivalence with
the RMT calculation.

We note that a preliminary outline of this theorem was published in
\cite{bk12}.  While in the process of writing up, we were notified by
Marcel Novaes that he achieved a similar breakthrough although only
for systems without time-reversal symmetry.  His original approach 
uses a combinatorial identity that remains unproved \cite{novaes12,novaes13}, 
but a new approach provides the complete equivalence using 
a semiclassical matrix model \cite{novaes13b}.

It is important to mention that Theorem~\ref{thm:main_equivalence}
does not produce any new formulae for moments such as $M_n(X)$, it
only establishes the equivalence of the two existing approaches to
their evaluation.  However, the description of the diagrams that we
develop in the proof can be analyzed further to yield new results.
This is done in the second part of this work \cite{bk13b}, where we
formalize the semiclassical evaluation of moments, re-cast it as a
summation over factorizations of given permutations, and thus
calculate $M_n(X)$ for any $n$ to several orders in the small
parameter $1/N$.

The layout of the paper is as follows: in sections \ref{sec:RMT} and
\ref{sec:semiclassics} we review the relevant notions and results from
random matrix theory and semiclassical approximation,
correspondingly.  Section~\ref{sec:restrictions} discusses general
properties of the sets of trajectories contributing on the
semiclassical side.  In section~\ref{sec:unitary} and
\ref{sec:orthogonal} we define and study the diagrams that classify
contributing sets of orbits and relate them to the random matrix
expansions for Circular Unitary and Orthogonal Ensembles,
correspondingly.  We endeavored to make the paper readable to a wide
variety of audiences, providing numerous examples and figures to
illustrate the discussed concepts.

%%%%%%%%%%%%%%%%%%%%%%%%%%%%%%%%%%%%%%%%%%%%%%%
\section{RMT prediction}
\label{sec:RMT}

It has been argued \cite{BluSmi_prl88,BluSmi_prl90} that scattering
through a chaotic cavity is described by a unitary matrix from an
appropriate random matrix ensemble.  In the absence of Time Reversal
Symmetry (TRS) the Circular Unitary Ensemble (CUE) is used.  If TRS 
is present the appropriate ensemble is the Circular Orthogonal Ensemble (COE).
The final classical symmetry class involves particles with spin $\frac{1}{2}$, 
for which breaking spin-rotation symmetry through spin-orbit interactions while 
retaining TRS leads to the Circular Symplectic Ensemble (CSE).

The CUE is the unitary group $U(N)$ endowed with
the Haar measure.  The averages of products of the elements of
matrices $U\in U(N)$ have been studied in
\cite{bb96,Mel_jpa90,Sam_jmp80,Wei_jmp78} among other works.  The main
result is
\begin{equation}
  \label{eq:avU}
  \langle U_{a_1a_\wb{1}}\ldots U_{a_sa_\wb{s}} 
  U^*_{b_1b_\wb{1}}\ldots U^*_{b_tb_\wb{t}} \rangle_{\mathrm{CUE}(N)} = 
  \delta_{t,s} \sum_{\sigma,\pi\in S_t} V^{\U}_N(\sigma^{-1}\pi)
  \prod_{k=1}^{t} \delta\!\left({a_k}-{b_{\sigma(k)}}\right) 
  \, \delta\!\left({a_\wb{k}}-{b_\wb{\pi(k)}}\right),
\end{equation}
where $S_t$ is the symmetric group of permutations of the set
$\{1,\ldots, t\}$, $\delta_{k,n} = \delta(k-n)$ is the Kronecker delta
(the latter notation is used solely to avoid nesting subindices) and
the coefficient $V^{\U}_N(\sigma^{-1}\pi)$ depends only on the lengths
of cycles in the cycle expansion of $\sigma^{-1}\pi$, i.e.\ on the
conjugacy class of the permutation $\sigma^{-1}\pi$.  For this reason
we will refer to $V^{\U}_N$ as the CUE class
coefficients.\footnote{Another name present in the literature is
  ``Weingarten'' function \cite{Wei_jmp78}, even though it was
  probably Samuel \cite{Sam_jmp80} who first defined the function and
  systematically studied it.}

The COE is the ensemble of \emph{unitary symmetric
  matrices}\footnote{Thus, despite the word ``orthogonal'' in the
  name, it is not the orthogonal group $O(N)$.  Rather, it can be
  identified with $U(N)/O(N)$.} with a probability distribution
obtained from the CUE through the mapping $W = UU^{\T}$, where $U$ is
a unitary matrix from the CUE and $U^{\T}$ its transpose.  The result
analogous to \eqref{eq:avU} reads \cite{bb96} (see also
\cite{MelSel_npa80}),
\begin{equation}
  \label{eq:avO}
 \langle W_{a_1a_\wb{1}}\ldots W_{a_sa_\wb{s}} 
  W^*_{b_1b_\wb{1}}\ldots W^*_{b_tb_\wb{t}} \rangle_{\mathrm{COE}(N)} = 
  \delta_{t,s} \sum_{\pi\in S_{2t}} V^{\OO}_N(\pi)
  \prod_{z\in Z_t}  \delta({a_z}-{b_{\pi(z)}}),
\end{equation}
where $\pi$ is a permutation on the set $Z_t = \{1,\ldots, t, \wb{t}, \ldots, \wb{1}\}$.

We mention that averaging formulae similar to \eqref{eq:avU} and
\eqref{eq:avO} have recently become available for a much bigger variety of
random matrix ensembles, see \cite{Mat_prep13} and references therein.

%%CSE
%The CSE is the ensemble of \emph{unitary self-dual matrices} and averaging over 
%the ensemble can be translated into an average over the COE.  The translation is given 
%by a simple set of rules derived in \cite{bb96} so that the results involve the same COE 
%class coefficients $V^{\OO}_N$.

Both types of class coefficients used above can be calculated recursively.
Namely, the class coefficients $V^{\U}_N$ were derived by Samuel
\cite{Sam_jmp80} to satisfy $V^{\U}_N(\emptyset) = 1$ and
\begin{multline}
  \label{eq:recur_CUE}
  N V_N^{\U}(c_1, \ldots, c_k) + \sum_{p+q=c_1} V_N^{\U}(p,q,c_2,\ldots, c_k)
  + \sum_{j=2}^{k} c_j V_N^{\U}(c_1+c_j,\ldots,\hat{c_j},\ldots, c_k) \\
  = \delta_{c_1,1} V_N^{\U}(c_2,\ldots, c_k) .
\end{multline}
Here $c_1,\ldots, c_k$ are the lengths of the cycles in the cycle expansion
of $\sigma^{-1}\pi$.  The notation $\hat{c_j}$ means that the element
$c_j$ has been removed from the list.  Finally, $\delta_{c_1,1}$ is
the Kronecker delta.

The corresponding recursion relation for the COE class coefficients
were derived by Brouwer and Beenakker \cite{bb96}.
They represent the permutation $\pi$ in equation \eqref{eq:avO} as the
product
\begin{equation}
  \label{eq:even_odd_rep}
  \pi = \Tt' \pi_e \pi_o \Tt'',
\end{equation}
where $\Tt'$ and $\Tt''$ are some involutions satisfying $\Tt^{\cdot}(j) =
j$ or $\wb{j}$, $\pi_o$ is a permutation on the set $\{1, \ldots, t\}$
and $\pi_e$ is a permutation on the set $\{\wb{t}, \ldots, \wb{1}\}$.
Note that factorization \eqref{eq:even_odd_rep} may be non-unique.
What is unique is the cycle structure of the permutation on $\{1,
\ldots, t\}$ defined by $\tau = \Tt \pi_e^{-1} \Tt \pi_o$, where
$\Tt=\left(1\wb{1}\right)\cdots\left(t\wb{t}\right)$ is used to ``cast''
the permutation $\pi_e$ into a permutation acting on $\{1, \ldots,
t\}$.

It turns out that the class coefficients $V^{\OO}_N(\pi)$ depend only on
the cycle type of the permutation $\tau$ defined above.  They satisfy
the recursion
\begin{multline}
  \label{eq:recur_COE}
  (N+c_1) V_N^{\OO}(c_1, \ldots, c_k) + \sum_{p+q=c_1} V_N^{\OO}(p,q,c_2,\ldots, c_k)
  + 2\sum_{j=2}^{k} c_j V_N^{\OO}(c_1+c_j,\ldots,\hat{c_j},\ldots, c_k) \\
  = \delta_{c_1,1} V_N^{\OO}(c_2,\ldots, c_k), 
\end{multline}
with the initial condition $V^{\OO}_N(\emptyset) = 1$.

In Lemma~\ref{lem:even_odd_and_orth_target} in Section~\ref{sec:target_orthogonal} 
we will give a simpler prescription for identifying the partition $c_1,\ldots,c_k$
which corresponds to a given $\pi$, bypassing the representation of
$\pi$ as the product in \eqref{eq:even_odd_rep}.

There are also expansions of $V^{\U}_N(\pi)$ in inverse powers of $N$
with coefficients expressed in terms of the number of factorizations
of $\pi$ of various types: primitive factorizations
\cite{MatNov_fpsac10}, inequivalent factorization \cite{BerIrv_prep},
and general factorizations \cite{Col_imrn03}.  It is the primitive
factorizations, discussed by Matsumoto and Novak
\cite{MatNov_fpsac10}, that will be particularly important to us.  We
will give an alternative proof of their result in
Section~\ref{sec:cancel_unit} and will extend it to the COE case in
Section~\ref{sec:cancel_orth}.

%%% Local Variables: 
%%% mode: latex
%%% TeX-master: "../ctsetm1"
%%% End: 

%%%%%%%%%%%%%%%%%%%%%%%%%%%%%%%%%%%%%%%%%%%%%%%%
\section{Semiclassical approximation}
\label{sec:semiclassics}

In this section we review the physical approximations involved in
evaluating correlations of the type
\begin{equation}
  \label{eq:correlator}
  C_E(\boldsymbol{a}, \boldsymbol{b}) = \langle \St_{a_1a_\wb{1}} \cdots \St_{a_sa_{\bar{s}}} 
  \St^*_{b_1b_\wb{1}} \cdots \St^*_{b_tb_{\bar t}} \rangle_{E}.
\end{equation}
First one employs the semiclassical approximation \cite{miller75,richter00,rs02} 
from equation \eqref{scatmatsemieqn}, leading to the expression
\begin{equation}
  \label{eq:correlator_semicl1}
  C_E(\boldsymbol{a}, \boldsymbol{b}) 
  = \left\langle \frac{1}{{(N\tD)}^{(s+t)/2}}
    \sum_{\substack{\gamma_j(a_{j}\to a_{\bar j}) \\
      \gamma'_k (b_{k} \to b_{\bar k} ) }}
    \prod_{j=1}^{s} A_{\gamma_j}\rme^{\frac{\rmi}{\hbar}S_{\gamma_j}(E)}
    \prod_{k=1}^{t}A_{\gamma'_k}^{*}\rme^{-\frac{\rmi}{\hbar}S_{\gamma'_k}(E)} 
  \right\rangle,
\end{equation}
here $N \tD$ is equal to the Heisenberg time. 

In the second step, which we will call the ``coinciding pathways
approximation'', it is argued that since $\hbar$ is very small, the
energy average is a sum of oscillatory integrals, so we first look 
for sets of trajectories that can have a small phase
\begin{equation*}
  \sum_{j} S_{\gamma_j}(E)-\sum_{k} S_{\gamma'_k}(E) \lesssim \hbar.
\end{equation*}
and contribute consistently in the semiclassical limit 
\cite{Sie02,sr01,Spe_jpa03,TurRic_jpa03}.  For each configuration 
of trajectories considered in the coinciding pathway approximation, 
in a final step the contribution is evaluated using further semiclassical 
approximations involving integrals over the full range of phases.

For each configuration considered, the union of paths of $\gamma_j$ in the phase space
must be almost identical to the union of paths $\gamma_k'$.  At this
point, TRS starts to play a role: if it is broken then the paths of 
$\gamma_k'$ must be traversed in the same direction, while in the 
presence of TRS the direction of traversal becomes irrelevant.

Let us consider the case $s=t=1$ for simplicity.  The easiest way to
achieve a small action difference is to let $\gamma = \gamma'$.  This
is known as the ``diagonal approximation'', pioneered for closed
systems by Berry \cite{berry85} using the sum rule of Hannay and 
Ozorio de Almeida \cite{ha84}, and for open systems that we consider
here by Bl\"umel and Smilansky \cite{BluSmi_prl88}.  However, it was
observed in \cite{bjs93a} that the diagonal approximation in some cases
fails to predict even the leading order contribution correctly.

\begin{figure}[t]
  \includegraphics{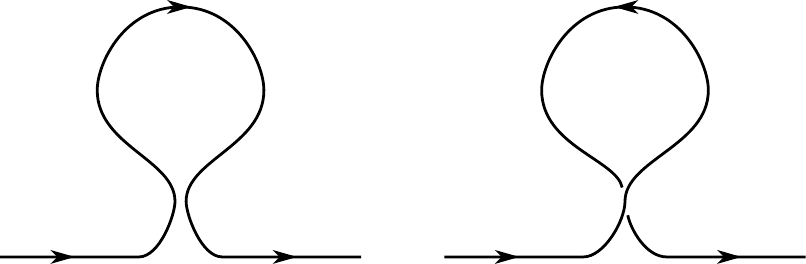}
  \caption{A schematic depiction of a pair of trajectories that provide the
    first off-diagonal contribution.  This pair requires TRS, 
    since the loop is traversed by the trajectories in opposite directions.} 
  \label{fig:SRpair}
\end{figure}

From analogy to disordered systems it was believed that the next
off-diagonal correction would come from trajectories $\gamma$ that
nearly intersect themselves, thus having a loop.  The partner
trajectory $\gamma'$ would run along $\gamma$ until the
self-intersection, then traverse the loop in the direction
\emph{opposite} to $\gamma$, and in the final part it would run along
$\gamma$ again.  This configuration requires TRS and is schematically depicted 
in Fig.~\ref{fig:SRpair}.  These general ideas were given analytical form 
in a breakthrough work by Richter and Sieber \cite{rs02} 
(see also \cite{Sie02,sr01}), who calculated the 
correction from the diagrams shown in Fig.~\ref{fig:SRpair}.  
This correction was evaluated by expressing the action 
difference, or phase, in terms of the angle of 
intersection of $\gamma'$ weighted by its ergodic average of 
occurrence and approximated by a stationary phase integral.  This
development paved the way for calculating higher order corrections and
higher order moments.   For example, one of the simplest configurations with
broken TRS is presented in Fig.~\ref{fig:NTRpair}.  Before we proceed, however, 
it is important to mention that the trajectory $\gamma$ can (and typically does) have
more than one ``near intersection''.  The diagram of
Fig.~\ref{fig:SRpair} is meant to represent the unique
near-intersection at which the trajectories $\gamma$ and $\gamma'$
go in different directions.

\begin{figure}[t]
  \includegraphics{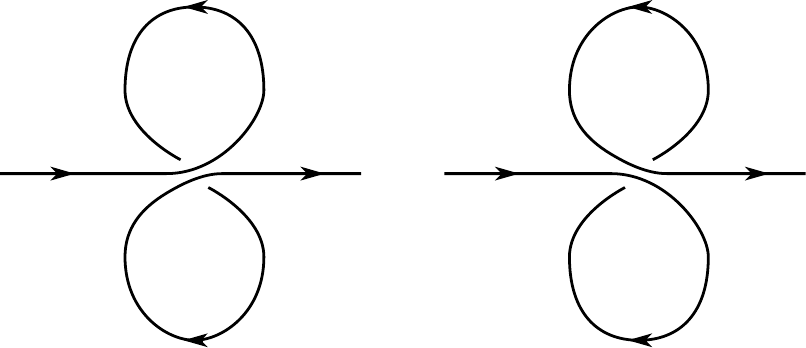}
  \caption{A simple pair of trajectories that do not require TRS: 
    both loops are traversed by both trajectories in the
    same direction, but in a different order.}
  \label{fig:NTRpair}
\end{figure}

The general idea is to split the set of paths into regions of two
types: ``links'' (or ``stretches'') that are traversed by exactly one
$\gamma_j$ and exactly one $\gamma_k'$ and ``encounters'' where
multiple stretches meet and the trajectories $\gamma$ interconnect
differently from the trajectories $\gamma'$.  Using additional assumptions
that all stretches are long and the sum over all
possible paths a stretch can take (between two given endpoints) is
well approximated by the ergodic average, M\"uller, Heusler,
Braun and Haake \cite{mulleretal07} formulated rules for
evaluating the contribution of a given topological arrangement of
trajectories.  This is the third major approximation involved in 
the Essen ansatz.  

\begin{definition}
  \label{def:Essen_ansatz}
  The total contribution of all pairs $\{\gamma_j\}$ and $\{\gamma_k'\}$
  with a given topological arrangement of links and encounters is
  given by a product where
  \begin{itemize}
  \item every link provides a factor $1/N$,
  \item every encounter of $2l$ stretches (an $l$-encounter) gives a
    factor of $-N$, 
  \item encounters that happen in the lead do not count (give a factor
    of $1$).
  \end{itemize}
\end{definition}

With these approximations the problem of evaluating any transport
statistic is reduced to the (hard) problem of counting all distinct topological
arrangements of links and encounters, so-called ``diagrams'', with
their respective weights calculated according to the rules above.

The rules of Definition~\ref{def:Essen_ansatz} ensure that the
contribution to a given correlator $C_E(\boldsymbol{a},
\boldsymbol{b})$ at a given order of $1/N$ comes from a finite number
of topological arrangements of trajectories.  To prove
Theorem~\ref{thm:main_equivalence} we will show that the sum over the
arrangements reproduces the RMT result at \emph{every} order of $1/N$.

%%% Local Variables: 
%%% mode: latex
%%% TeX-master: "../ctsetm1"
%%% End: 

%%%%%%%%%%%%%%%%%%%%%%%%%%%%%%%%%%%%%%%%%%%%%%%%%%%%
\section{Restrictions on matrix coefficients}
\label{sec:restrictions}

Within the coinciding pathways approximation, links of the
trajectories $\{\gamma_j\}$ should be in a one-to-one correspondence to
links of the trajectories $\{\gamma_k'\}$.  For the links starting
or ending in channels in the lead, this implies restriction on which
correlators from \eqref{eq:correlator} can be non-zero.  Here we
describe these restrictions, first for the systems with broken TRS
(the unitary case) as it is simpler and then for systems with TRS (the
orthogonal case).

%%%%%%%%%%%%%%%%%%%%%%%%%%%%%%%%%%%%%%%%%%%%%%%%%%%%
\subsection{The unitary case}
\label{sec:target_unitary}

For the unitary case with broken TRS, assume
we have a suitable configuration of trajectories $\gamma$ and
$\gamma'$, contributing to $C_E(\boldsymbol{a}, \boldsymbol{b})$.  We start
from channel $a_1$ and follow the trajectory $\gamma_1$.  The final
stretch of $\gamma_1$ leads to the channel $a_{\wb{1}}$.  With the coinciding pathways approximation, 
the same stretch is traversed by a $\gamma'$ trajectory, which we denote by
$\gamma_{\pi(1)}'$.  The trajectory $\gamma_{\pi(1)}'$, by definition,
goes between channels $b_{\pi(1)}$ and $b_{\wb{\pi(1)}}$.  The final
stretch of $\gamma_1$ must also be final for $\gamma_{\pi(1)}'$ (since it is leading to a
channel and not an encounter and since the trajectories must run in
the same direction).  We immediately conclude that
\begin{equation}
  \label{eq:pi-def}
  a_{\wb{1}} = b_{\wb{\pi(1)}}.
\end{equation}
We follow the trajectory $\gamma_{\pi(1)}'$ backwards,
until we are on it's first stretch, about to hit channel number
$b_{\pi(1)}$.  The partner of $\gamma_{\pi(1)}'$ on this
stretch is a $\gamma$ trajectory, which we will denote
$\gamma_{\tau(1)}$.  We
can now conclude that
\begin{equation}
  \label{eq:same_i_channels}
  a_{\tau(1)} = b_{\pi(1)}.
\end{equation}

We can now follow trajectory $\gamma_2$, finding the value of $\pi(2)$
and $\tau(2)$ etcetera.  It is clear that the thus defined functions
$\pi$ and $\tau$ are permutations.  Since the number of end-points of
the trajectories $\gamma$ needs to be the same as the number of
end-points of the trajectories $\gamma'$, we immediately get $s=t$ in
Eq.~(\ref{eq:correlator}).  Henceforth we denote this common value by
$n$.

Further, the permutation $\pi$ imposes $n$ restrictions on the sets
$\boldsymbol{a}$ and $\boldsymbol{b}$, namely,
\begin{equation*}
  a_{\bar j} = b_{\wb{\pi(j)}}, \qquad j=1\ldots n.
\end{equation*}
Letting $\sigma = \pi \tau^{-1}$ we can transform identities similar
to \eqref{eq:same_i_channels} into
\begin{equation*}
 a_{j} = b_{\sigma(j)},  \qquad j=1\ldots n.
\end{equation*}
The permutation $\tau$ defined above and calculated from $\sigma$ and
$\pi$ as $\tau = \sigma^{-1}\pi$ will be called the ``target
permutation''.

We can summarize our discussion as a lemma.

\begin{lemma}
  \label{lem:structure_of_correlator}
  Let $\boldsymbol{a} \in \mathbb{N}^s$, $\boldsymbol{b} \in \mathbb{N}^t$.
  Within the ``coinciding pathways approximation'', $C_E^{\U}(\boldsymbol{a},\boldsymbol{b})$
  is zero unless $s=t$ and there exist permutations $\pi, \sigma \in
  S_{t}$ such that
  \begin{equation}
    \label{eq:ab_relations}
    a_{j} = b_{\sigma(j)} \qquad \mbox{and} \qquad 
    a_{\wb{j}} = b_{\wb{\pi(j)}}, \qquad j=1\ldots t.
  \end{equation}
  Moreover, if two pairs, $(\pi_1,\sigma_1)$ and $(\pi_2,\sigma_2)$,
  both fulfill \eqref{eq:ab_relations} and have the same target
  permutation $\tau=\sigma_1^{-1}\pi_1 = \sigma_2^{-1}\pi_2$, then
  their contributions to $C_E^{\U}(\boldsymbol{a},\boldsymbol{b})$ are identical.  In
  other words,
  \begin{equation}
    \label{eq:correlator_sum_target_perm}
    C_E^{\U}(\boldsymbol{a},\boldsymbol{b}) = \delta_{t,s} \sum_{\sigma,\pi \in
      S_t} \Delta^{\U}(\sigma^{-1}\pi) \prod_{j=1}^{t} \delta\!\left({a_j}-{b_{\sigma(j)}}\right) 
    \, \delta\!\left({a_\wb{j}}-{b_\wb{\pi(j)}}\right),
  \end{equation}
  where $\Delta^{\U}(\sigma^{-1}\pi)$ is the total contribution of
  trajectories $\gamma$ and $\gamma'$ whose ends satisfy
  \eqref{eq:ab_relations}.
\end{lemma}

\begin{remark}
  The permutations $\sigma$ and $\pi$ have the same role as those
  appearing on the right-hand side of equation~\eqref{eq:avU}.
\end{remark}

\begin{proof}
  What remains to be proved in Lemma~\ref{lem:structure_of_correlator} is
  that the contributions of $(\pi_1,\sigma_1)$ and $(\pi_2,\sigma_2)$
  are identical.  We will exhibit a one-to-one correspondence between
  the set of trajectories with ends satisfying
  \begin{equation}
    \label{eq:ab_relations1}
    a_{j} = b_{\sigma_1(j)} \qquad \mbox{and} \qquad 
    a_{\bar j} = b_{\widebar{\pi_1(j)}}, \qquad j=1\ldots n
  \end{equation}
  and trajectories with ends satisfying
  \begin{equation}
    \label{eq:ab_relations2}
    a_{j} = b_{\sigma_2(j)} \qquad \mbox{and} \qquad 
    a_{\bar j} = b_{\widebar{\pi_2(j)}}, \qquad j=1\ldots n.
  \end{equation}
  To do so we simply relabel the trajectories $\gamma'$, so that the
  trajectory $\gamma_j'$ becomes the trajectory
  $\widetilde{\gamma}'_{\pi_2\pi_1^{-1}(j)}$.  Then the second
  identity in \eqref{eq:ab_relations1} becomes
  \begin{equation*}
    a_{\bar j} = b_{\widebar{\pi_2\pi_1^{-1}\pi_1(j)}} = b_{\widebar{\pi_2(j)}},
  \end{equation*}
  and the first one
  \begin{equation*}
    a_{j} = b_{\pi_2\pi_1^{-1}\sigma_1(j)} =
    b_{\pi_2\pi_2^{-1}\sigma_2(j)} = b_{\sigma_2(j)},
  \end{equation*}
  where the equality of the target permutations was used.
\end{proof}

\begin{example}
  Consider the correlator
 \begin{equation*}
    \langle \St_{1,2} \St^*_{2,1} \rangle_E, \qquad \mbox{i.e.}\quad
    a_1=1,\ a_\wb{1}=2, \quad b_1 = 2, \ b_\wb{1} = 1.
  \end{equation*}
  For a system with broken time reversal symmetry this correlator is
  zero.  Indeed, there are no permutations $\sigma$ and $\pi$ (on 1
  element) that can satisfy equations \eqref{eq:ab_relations}.
\end{example}

\begin{example}
  \label{ex:3traj}
  Consider the correlator
  \begin{equation*}
    \langle \St_{1,2} \St_{3,4} \St_{5,6} \, \St^*_{3,6} \St^*_{5,4} \St^*_{1,2}\rangle_E.
  \end{equation*}
  The trajectories run
  \begin{align*}
    &\gamma_1:1\to 2, & &\gamma_2:3\to4, & &\gamma_3:5\to6,\\
    &\gamma_1':3\to6, & &\gamma_2':5\to4, & &\gamma_3':1\to2,
  \end{align*}
  and the variables $a$ and $b$ have values
  \begin{align*}
    &a_1 = 1,& &a_{\wb{1}}=2,& &a_2 = 3,& &a_{\wb{2}}=4,& 
    &a_3 = 5,& &a_{\wb{3}}=6, \\
    &b_1 = 3,& &b_{\wb{1}}=6,& &b_2 = 5,& &b_{\wb{2}}=4,& 
    &b_3 = 1,& &b_{\wb{3}}=2.
  \end{align*}
  Two examples of trajectory configurations contributing to the
  correlator above are given in Fig.~\ref{fig:NTRtraj_adv}.  Note that
  the start and end of each trajectory (unfilled circles) are labelled
  by the index of the corresponding variable $a_j$ or $a_{\bar{j}}$
  and not by its value.

 From equation~\eqref{eq:ab_relations} it is clear that the only
  choice for $\sigma$ and $\tau$ is (in the two-line notation)
  \begin{equation*}
    \sigma = 
    \begin{pmatrix}
      1 & 2 & 3\\
      3 & 1 & 2
    \end{pmatrix},
    \qquad
    \pi = 
    \begin{pmatrix}
      1 & 2 & 3\\
      3 & 2 & 1
    \end{pmatrix}.
  \end{equation*}
  Therefore, the target permutation is $\tau = (1)(2\,3)$ (in the
  cycle notation).  The target permutation can be read off the diagram
  (see Fig.~\ref{fig:NTRtraj_adv}) by starting at $j$, then following
  the $\gamma$ trajectory until its end, then following the $\gamma'$
  trajectory in reverse to its start.  The label there is the image
  of $j$ under the action of $\tau$.
\end{example}

\begin{figure}[t]
  \includegraphics[scale=0.7]{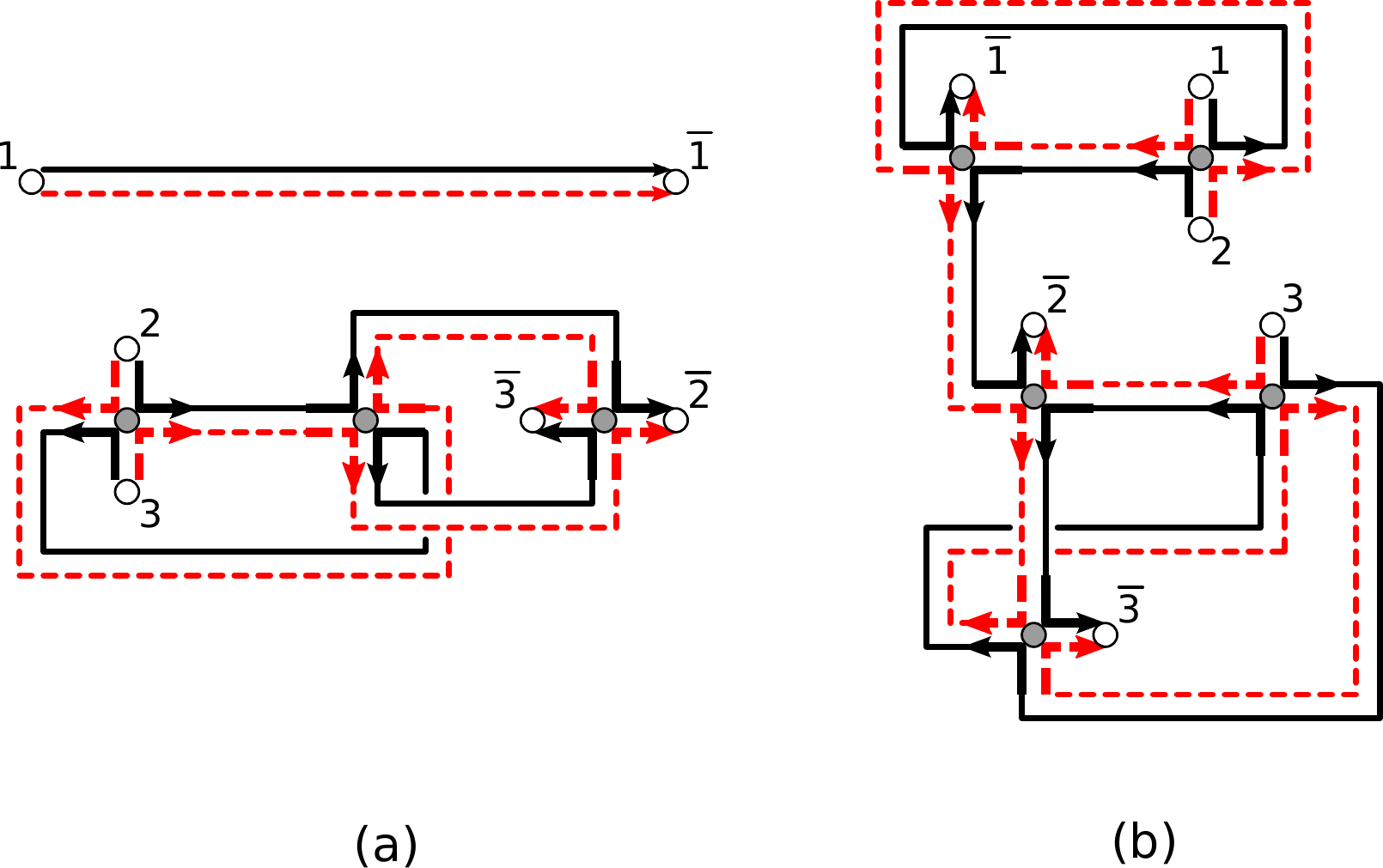}
  \caption{Two examples of configurations of trajectories that
    contribute to the correlator in Example~\ref{ex:3traj}.  The
    endpoints of trajectories are labelled by the indices of the
    variables $a_1,\ldots,a_3$ and $a_\wb{1},\ldots,a_\wb{3}$.  The
    trajectories $\gamma_j$ are drawn in black solid
    lines; they go from $j$ to $\wb{j}$.  The trajectories $\gamma_j'$
    are drawn in red dashed line and go from $\tau(j)$ to $\wb{j}$.}
  \label{fig:NTRtraj_adv}
\end{figure}

\begin{example}
  For the correlator
  \begin{equation*}
    \langle \St_{1,2} \St_{3,2} \, \St^*_{1,2} \St^*_{3,2}\rangle_E
  \end{equation*}
  there are two choices of the pair $(\sigma, \pi)$.  The mapping
  $\sigma$ between the first indices is the identity permutation $id$,
  while the mapping between the second indices can either be $id$ or
  $(1\,2)$.  Therefore the target permutation $\tau$ is
  \begin{equation*}
    \tau = id \qquad \mbox{or} \qquad \tau=(1\,2).
  \end{equation*}
\end{example}

\begin{remark}
  \label{rem:moment_correlator}
  Consider the correlator of the type
  \begin{equation*}
    \left\langle \St_{i_1,o_1} \St_{i_2,o_2} \cdots \St_{i_n,o_n}
    \St^*_{i_2,o_1} \St^*_{i_3,o_2} \cdots \St^*_{i_1,o_n} \right\rangle,
  \end{equation*}
  which arise in evaluation of moments \eqref{eq:trace_expansion}.
  Irrespective of the choices of the values for the indices
  $i_1,\ldots o_n$, one can always let
  \begin{equation*}
    \sigma = (n \ldots 2\, 1), \qquad \pi = id.
  \end{equation*}
  leading to $\tau$ being the grand cycle, $\tau=(1\,2\ldots n)$.  In
  previous papers, starting with \cite{bhn08}, the diagrams realizing
  this choice of $\tau$ were considered as the base contribution.
  Correcting factors were used to take care of other target
  permutations, arising, for example, when $i_j = i_k$ for some $j$
  and $k$.  It was argued that other target permutations corresponded
  to encounters happening in a lead (which essentially removes the
  encounter from the diagram, see the rules of the Essen ansatz).  An
  encounter in the base diagram that \emph{could} be moved into the
  lead can be seen as ``untying'' the encounter.  We consider this
  interpretation in detail in the second half of this work
  \cite{bk13b} since it provides an easier way to obtain answers for
  the moment generating functions.

  For the proof of our main theorem, however, it is more convenient to
  have each diagram representing one target permutation and to sum
  over the target permutation, as is already done on the RMT side.  
  The work by Novaes \cite{novaes12,novaes13} adopts a similar strategy.
\end{remark}

%%%%%%%%%%%%%%%%%%%%%%%%%%%%%%%%%%%%%%%%%%%%%%%%%%
\subsection{The orthogonal case}
\label{sec:target_orthogonal}

In the presence of TRS there are more possibilities to match
parts of the trajectories $\gamma$ to $\gamma'$, since now the
``head'' can be matched with the ``tail''.

More precisely, the channel $a_1$ lies at the start of the trajectory
$\gamma_1$ but must also lie on a trajectory $\gamma_j'$.  Therefore
$a_1$ must coincide either with $b_{j}$ or $b_\wb{j}$.
In other words,
\begin{equation*}
  a_1 = b_{\varpi(1)},
\end{equation*}
where $\varpi(1) \in \{1, \ldots, t, \wb{t}, \ldots, \wb{1}\}$.
Similarly for all other channels $a_j$ there is a matching channel
$b_{\varpi(j)}$, where $j$ can be $1,\ldots, t$ or
$\wb{t},\ldots,\wb{1}$.  This defines the permutation $\varpi$ on $2t$
symbols $Z_t = \{1, \ldots, t, \wb{t}, \ldots, \wb{1}\}$.  Thus the
only restriction on the indices $a_j$ and $b_j$ is that they are
equal as multi-sets (i.e. contain the same elements the same number of
times).

\begin{notation}
  \label{not:labels}
  In the description of orthogonal trajectories, we adopt the
  convention that $j$ or $k$ refers to the variable label that does
  not have the bar (correspondingly, $\wb{j}$ is a label that does
  have the bar), while $z$ denotes a label either with or without the
  bar.
\end{notation}

To understand the analogue of the \emph{target permutation} in the
orthogonal case, we take another look at the unitary case.  The target
permutation was $\tau = \sigma^{-1} \pi$.  However, it could be argued
that $\pi$ and $\sigma$ act on different spaces: $\sigma$ acts on the
elements $\{1,\ldots, t\}$ and $\pi$ acts on $\{\wb{1},\ldots,
\wb{t}\}$.  To multiply the permutations we need to map them onto the
same space, for example using the mapping $\Tt: j \mapsto \wb{j}$.
Then the ``correct'' expression for the target permutation is $\tau =
\sigma^{-1} \Tt^{-1} \pi \Tt$.

This correction, somewhat superfluous in the unitary case, becomes a
necessity in the orthogonal case.  The meaning of the mapping $\Tt$ is
``propagation along the trajectory'' $\gamma$ or $\gamma'$.  We define
it as a permutation on $Z_t$:
\begin{equation*}
  \Tt = (1\,\wb{1})\cdots(t\,\wb{t}).
\end{equation*}
In particular, $\Tt$ is an involution, i.e. $\Tt^{-1}=\Tt$.  Moreover, in
the orthogonal case, the permutation $\varpi$ plays the role of $\pi$
and $\sigma$ \emph{combined}.  We thus define the target permutation
$\tau$ by 
\begin{equation}
  \label{eq:orth_target}
  \tau = \varpi^{-1} T^{-1} \varpi T.
\end{equation}
It acts on the ends of trajectories $\gamma$ in the following fashion:
take an end, propagate along $\gamma$ to the other end, find the
corresponding end of a trajectory $\gamma'$, propagate to its other
end.  This is the corresponding end of the next $\gamma$-trajectory.

\begin{example}
  \label{ex:TR_corr}
  Consider the correlator $\langle \St_{12} \St_{34} \St_{56} \,
  \St^{*}_{54} \St^{*}_{31} \St^{*}_{62} \rangle$.  The permutation
  $\varpi$ (in two-row notation) is
  \begin{equation*}
    \varpi =
    \begin{pmatrix}
      1 & \wb{1} & 2 & \wb{2} & 3 & \wb{3}\\
      \wb{2} & \wb{3} & 2 & \wb{1} & 1 & 3
    \end{pmatrix}.    
  \end{equation*}
  Then the target permutation is $\tau =
  (1\,\wb{3}\,\wb{2})(2\,3\,\wb{1})$.  
  Two examples of trajectory configurations contributing to the
  correlator above are given in Fig.~\ref{fig:TRtraj_adv}
\end{example}

\begin{example}
  \label{ex:TR_corr_2}
  Consider the correlator $\langle \St_{12} \St_{31} \, \St^{*}_{23}
  \St^{*}_{11} \rangle$.  There are two possibilities for the
  permutation $\varphi$, namely
  \begin{equation*}
    \varpi =
    \begin{pmatrix}
      1 & \wb{1} & 2 & \wb{2} \\
      2 & 1 & \wb{1} & \wb{2}
    \end{pmatrix}
    \qquad\mbox{and}\qquad
    \varpi =
    \begin{pmatrix}
      1 & \wb{1} & 2 & \wb{2} \\
      \wb{2} & 1 & \wb{1} & 2
    \end{pmatrix}.
  \end{equation*}
  Both of these correspond to the same target permutation
  $\tau=(1\,2)(\wb2\,\wb1)$.
\end{example}

Note the ``palindromic'' symmetry of the results of
Examples~\ref{ex:TR_corr} and \ref{ex:TR_corr_2}.  This observation is
made precise in the following Lemma.

\begin{figure}
  \includegraphics[scale=0.7]{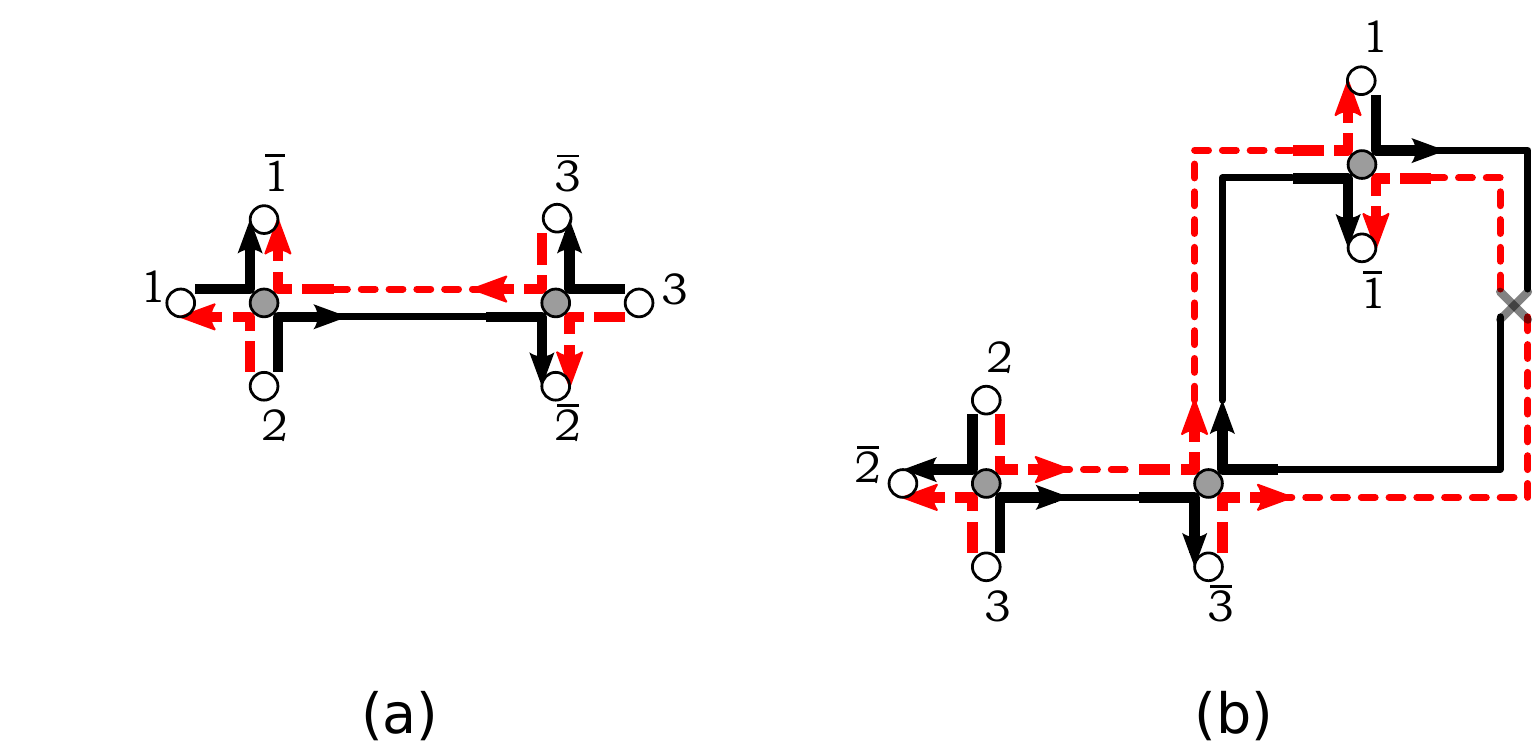}
  \caption{Two examples of trajectory configurations contributing to
    the correlator of Example~\ref{ex:TR_corr}.  In the second
    diagram, to draw the trajectories on the plane we had to allow
    them to cross.}
  \label{fig:TRtraj_adv}
\end{figure}

\begin{lemma}
  \label{lem:orth_target_prop}
  The orthogonal target permutation has the following properties:
  \begin{enumerate}
  \item \label{enum:sym} if $\tau(x) = y$ then $\tau(\wb{y}) = \wb{x}$,
  \item \label{enum:distinct} $\tau(x) \neq \wb{x}$.
 \end{enumerate}
  Therefore the cycles on $\tau$ come in symmetric pairs: for every
  cycle $(z_1\, z_2 \, z_3 \ldots)$ there is the \emph{distinct} partner
  cycle $(\ldots\, \wb{z_3} \, \wb{z_2} \, \wb{z_1})$.
\end{lemma}

\begin{proof}
  Denote $\varpi(\wb{x}) =: z$.  Then we have $y = \tau(x) =
  \varpi^{-1}T^{-1} \varpi(\wb{x}) = \varpi^{-1}(\wb{z})$.  Therefore $\varpi(y) =
  \wb{z}$.

  Now we can calculate $\tau(\wb{y}) = \varpi^{-1}T^{-1} \varpi(y) =
  \varpi^{-1}(z) = \wb{x}$.

  For the second property, we assume the contrary: $\tau(x) =
  \wb{x}$.  Then $\varpi^{-1} T^{-1} \varpi(\wb{x}) = \wb{x}$ or 
  the contradiction
  \begin{equation*}
    \varpi(\wb{x}) = \wb{ \varpi(\wb{x}) }.
  \end{equation*}

  The conclusion now follows.  Indeed the first property implies that
  for every cycle $(z_1\, z_2 \, z_3 \ldots)$ there is the
  cycle $(\ldots\, \wb{z_3} \, \wb{z_2} \, \wb{z_1})$.  The fact that
  these are not the same cycle follows from the second property.
\end{proof}

\begin{remark}
  \label{rem:Ttau_involution}
  Property \ref{enum:sym} of Lemma~\ref{lem:orth_target_prop} can be
  written as $\tau(\Tt \tau(x)) = \wb{x}$ or, equivalently, $\Tt\tau
  \Tt\tau = id$.  In other words, $\Tt\tau$ is an involution.  Property
  \ref{enum:distinct} is equivalent to $\Tt\tau(x) \neq x$.  Therefore,
  Lemma~\ref{lem:orth_target_prop} can be reformulated as saying that
  $\Tt\tau$ consists only of cycles of length 2.
\end{remark}

\begin{remark}
  An equivalent way to express the symmetry of $\tau$ is via the
  identity $\tau( x ) = \wb{ \tau^{-1}(\wb{x}) }$.  This is similar to
  the definition of the \emph{hyperoctahedral group} that is the
  subgroup of $S_{2t}$ with elements satisfying $\chi(x) =
  \wb{\chi(\wb{x})}$.  However, our target permutations are different,
  since they do not form a group.
\end{remark}

There is a close connection between the target permutation $\tau$ and
the cycle structure appearing as the true parameters of the class
coefficients $V^{\OO}_N$, see (\ref{eq:avO}).

\begin{lemma}
  \label{lem:even_odd_and_orth_target}
  Let the permutation $\varpi$ be represented as the product $\Tt'
  \pi_o \pi_e \Tt''$, where $\Tt'$ and $\Tt''$ are some involutions
  that contain only cycles of the form $(j)(\wb{j})$ or
  $(j\,\wb{j})$, $\pi_o$ is a permutation on the set $\{1, \ldots,
  t\}$ and $\pi_e$ is a permutation on the set $\{\wb{t}, \ldots,
  \wb{1}\}$, see Section~\ref{sec:RMT}.

  Then the cycle type of $\tau = \varpi^{-1} \Tt^{-1} \varpi \Tt$ is twice
  that of the permutation $\Tt \pi_e^{-1} \Tt \pi_o$.  Consequently, the
  latter cycle type does not depend on the choice of the representation.
\end{lemma}

\begin{proof}
  Substituting the representation $\varpi = \Tt'
  \pi_o \pi_e \Tt''$ into the definition of $\tau$ and noting that the
  involutions $\Tt'$, $\Tt''$ and $\Tt$ commute, we get
  \begin{align*}
    \tau = \Tt \left(\Tt'' \pi_e^{-1} \pi_o^{-1} \Tt'\right) \Tt \left(\Tt' \pi_o \pi_e \Tt''\right) 
    = \Tt'' \left(\Tt  \pi_e^{-1} \Tt\right) \left(\Tt \pi_o^{-1} \Tt\right) \pi_o \pi_e \, \Tt''
    \sim \left(\Tt \pi_e^{-1} \Tt \pi_o\right) \left(\Tt \pi_o^{-1} \Tt \pi_e\right), 
  \end{align*}
  where $\sim$ denotes conjugate permutations.  Obviously, the two
  parts of the last expression act on disjoint sets $\{1, \ldots, t\}$
  and $\{\wb{t}, \ldots, \wb{1}\}$ and have the same cycle type, which
  explains the doubling.
\end{proof}

A result analogous to Lemma~\ref{lem:structure_of_correlator}
summarizes the discussion of this section.

\begin{lemma}
  \label{lem:structure_of_correlator_O}
  Let $\boldsymbol{a} \in \mathbb{N}^s$, $\boldsymbol{b} \in \mathbb{N}^t$.  Within
  the ``coinciding pathways approximation'', $C_E^{\OO}(\boldsymbol{a},\boldsymbol{b})$
  is zero unless $s=t$ and there exists a permutation $\varpi$ on
  $Z_t = \{1,\ldots, t, \wb{t}, \ldots, \wb{1}, \}$ such that
  \begin{equation}
    \label{eq:ab_relations_O}
    a_{z} = b_{\varpi(z)}, \qquad z \in Z_t.
  \end{equation}
  Moreover, if two permutations, $\varpi_1$ and $\varpi_2$,
  both fulfill \eqref{eq:ab_relations_O} and have the same target
  permutation $\tau=\varpi_1^{-1} \Tt^{-1} \varpi_1 \Tt
  =\varpi_2^{-1} \Tt^{-1} \varpi_2 \Tt$, then
  their contributions to $C_E^{\OO}(\boldsymbol{a},\boldsymbol{b})$ are identical.  In
  other words,
  \begin{equation}
    \label{eq:correlator_sum_target_perm_O}
    C_E^{\OO}(\boldsymbol{a},\boldsymbol{b}) = \delta_{t,s} \sum_{\varpi \in
      S_{2t}} \Delta^{\OO}(\tau) \prod_{z\in Z_t}^{t} \delta\!\left({a_z}-{b_{\varpi(z)}}\right),
  \end{equation}
  where $\Delta^{\OO}(\tau)$ is the total contribution of trajectories
  $\gamma$ and $\gamma'$ whose ends satisfy \eqref{eq:ab_relations_O}.
\end{lemma}

\begin{proof}
  As before, the only part we need to prove is that the contributions
  of trajectories satisfying (\ref{eq:ab_relations_O}) with $\varpi_1$
  and $\varpi_2$ are identical if the corresponding target
  permutations coincide.  Starting with a set of trajectories
  described by $\varpi_1$ we will relabel them so that they will be
  described by $\varpi_2$.  We will only relabel the partner
  trajectories $\gamma'$.  Note that direction-reversal is allowed.

  The relabelling permutation we will denote by $\rho$.  It acts on the
  ends of the trajectories, namely if the original trajectory $\gamma_j'$ run
  from $b_{j}$ to $b_{\wb{j}}$, the new one $\wt{\gamma}_j'$ runs from
  $\wt{b}_j = b_{\rho(j)}$ to $\wt{b}_{\wb{j}} = b_{\rho(\wb{j})}$.

  We want to have $a_z = \wt{b}_{\varpi_2(z)}$.  On one hand we have
  $a_z = b_{\varpi_1(z)}$ (since the original trajectories agreed
  with $\varpi_1$).  On the other, we have $\wt{b}_{\varpi_2(z)} =
  b_{\rho(\varpi_2(z))}$.  Therefore, $\rho$ needs to satisfy
  $\rho\varpi_2 = \varpi_1$, or
  \begin{equation*}
    \rho = \varpi_1 \varpi_2^{-1}.
  \end{equation*}

  We now verify that $\rho$ is a valid relabelling.  Namely, if it
  maps the trajectory end $j$ to the trajectory end $k$ (or $\wb{k}$),
  then it must map the other end $\wb{j}$ to $\wb{k}$
  (correspondingly, $k$).
  Putting it formally, $\rho$ must satisfy $\rho \Tt = \Tt \rho$.  By the
  definition of the target permutation (and since $\Tt$ is an
  involution), we have $\varpi_2^{-1}\Tt = \tau \Tt \varpi_2^{-1}$.
  Therefore,
  \begin{equation*}
    \rho \Tt = \varpi_1 \varpi_2^{-1} \Tt 
    = \varpi_1 \tau \Tt \varpi_2^{-1}
    = \varpi_1 \left(\varpi_1^{-1} \Tt \varpi_1 \right) \varpi_2^{-1}
    = \Tt \varpi_1 \varpi_2^{-1} = \Tt \rho,
  \end{equation*}
  as desired.
\end{proof}

%%% Local Variables: 
%%% mode: latex
%%% TeX-master: "../ctsetm1"
%%% End: 

\section{Summation over the unitary diagrams}
\label{sec:unitary}

Diagrams are schematic depictions of sets of trajectories, describing
which parts of trajectories $\{\gamma_j\}$ follow which parts of
trajectories $\{\gamma_j'\}$ by using a graph whose edges correspond
to stretches of $\gamma$ and vertices correspond to encounters (see
section~\ref{sec:semiclassics}).

In this section we describe how the diagrams can be encoded
mathematically, operations acting within the set of diagrams and the
cancellations resulting when we evaluate the contributions $\Delta^U$
as a sum over all possible diagrams.  We will consider here the case
of broken TRS (``unitary'' diagrams) before treating TRS in
section~\ref{sec:orthogonal}.

%%%%%%%%%%%%%%%%%%%%%%%%%%%%%%%%%%%%%%%%%%%%%%%%%%
\subsection{Diagrams as ribbon graphs}

The drawings in Fig.~\ref{fig:NTRtraj_adv} hint that the natural
mathematical description of the topology of a set of contributing
trajectories is a \emph{ribbon graph}, also known as a \emph{fat
  graph} or a \emph{map}.  In combinatorics, a \emph{map} is a graph
which is drawn on a surface in such a way that (a) the edges do not
intersect and (b) cutting the surface along the edges one obtains
disjoint connected sets (``countries''), each homeomorphic to an open
disk.  The act of drawing a graph on surface fixes a cyclic ordering
of the edges incident to each vertex.  For an accessible introduction
to combinatorial maps and their application to RMT
(which is only one of their numerous applications) we refer the reader
to \cite{Zvo_mcm97}.  Further information can be found in
\cite{JacVis_Atlas,Tutte_GraphTheory}, among other sources.  Note that
in the present article we only use the term ``map'' in the
combinatorial sense.

Lacking a surface of sufficiently high genus, maps can be drawn on a
sheet of paper (though the edges might need to intersect).  It is convenient to
fatten the edges of such graphs to make them \emph{ribbons}, with two sides,
and represent the maps as \emph{ribbon graphs}.
At vertices, which are also fattened, the side of one edge is
connected to a side of another, which corresponds to the cyclic
ordering mentioned above.  Following the sides until we return to the
starting spot we trace out \emph{boundary walks} of the map.

In semiclassical diagrams, each trajectory stretch corresponds to a
fattened edge of the graph.  The trajectory $\gamma$ runs along one
side of the edge and $\gamma'$ runs (in the same direction) along the
other side.  The vertices of degree 1, henceforth called \emph{leaves},
correspond to trajectories either starting from or exiting into a
channel.  They are labelled by the symbols $1,\ldots,t$ and
$\wb{1},\ldots,\wb{t}$, which are understood to correspond to the
indices of the variables $\boldsymbol{a}$ and $\boldsymbol{b}$ in the
correlator $C_E(\boldsymbol{a},\boldsymbol{b})$,
equation~\eqref{eq:correlator}.  All other vertices correspond to
encounters which we will refer to as \emph{internal} vertices.  At a
vertex the trajectory can be followed by continuing along the side of
the edge onto the side of the vertex and then on to the corresponding
side of the next edge.

\begin{figure}[t]
  \includegraphics[scale=0.7]{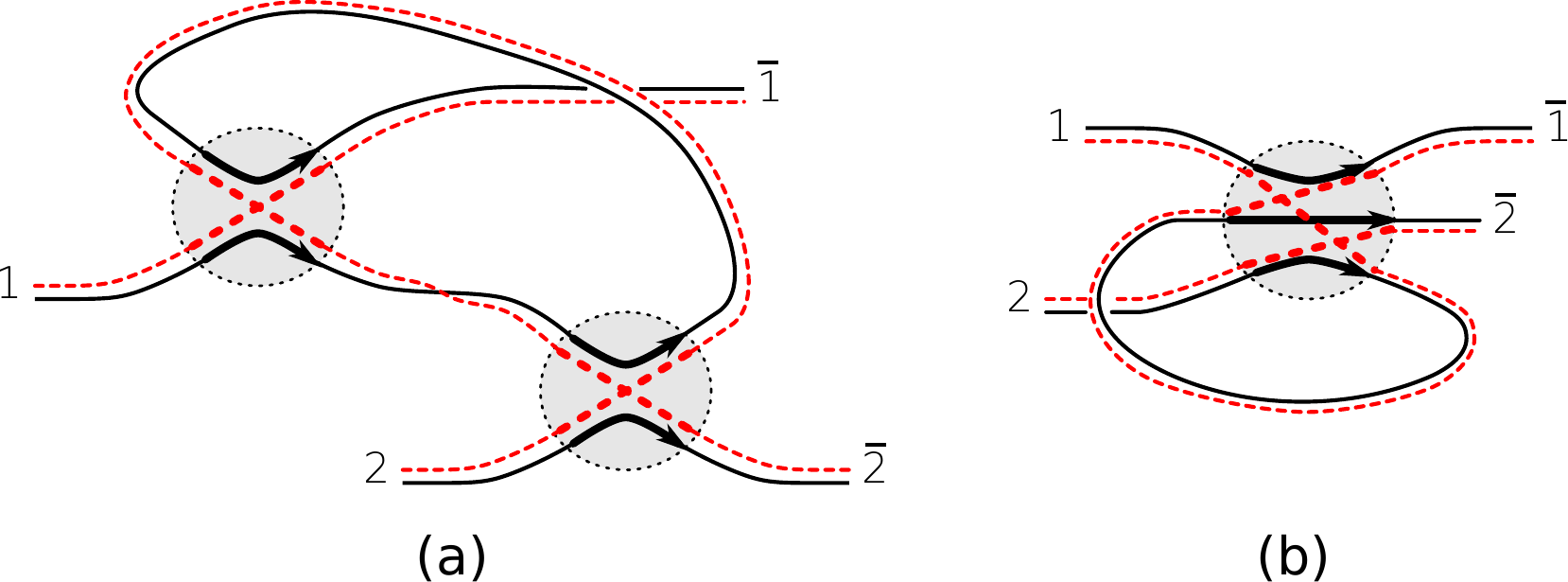}
  \caption{Two examples of semiclassical diagrams as drawn in the
    physics literature.  These examples correspond to diagrams (d) and
    (e) of Fig.~4 in \cite{mulleretal07}}
  \label{fig:encounter_physics}
\end{figure}

The ribbon graphs used in the present article are not significantly
different from the more traditional way the semiclassical diagrams are
depicted in the physics literature (see,
e.g.\ \cite{BerSchWhi_jpa03,mulleretal07,mulleretal05,SmiLerAlt_prb98}).
In the ``physics'' diagrams, the trajectories are shown as smooth with
the reconnections (which distinguish trajectories $\gamma$ from trajectories
$\gamma'$) happening in the small encounter regions, see
Fig.~\ref{fig:encounter_physics}.  Within
the encounter region trajectories are drawn as intersecting.  This has
the unfortunate effect of confining a lot of information to a small
part of the drawing.

\begin{figure}[t]
  \includegraphics[scale=0.7]{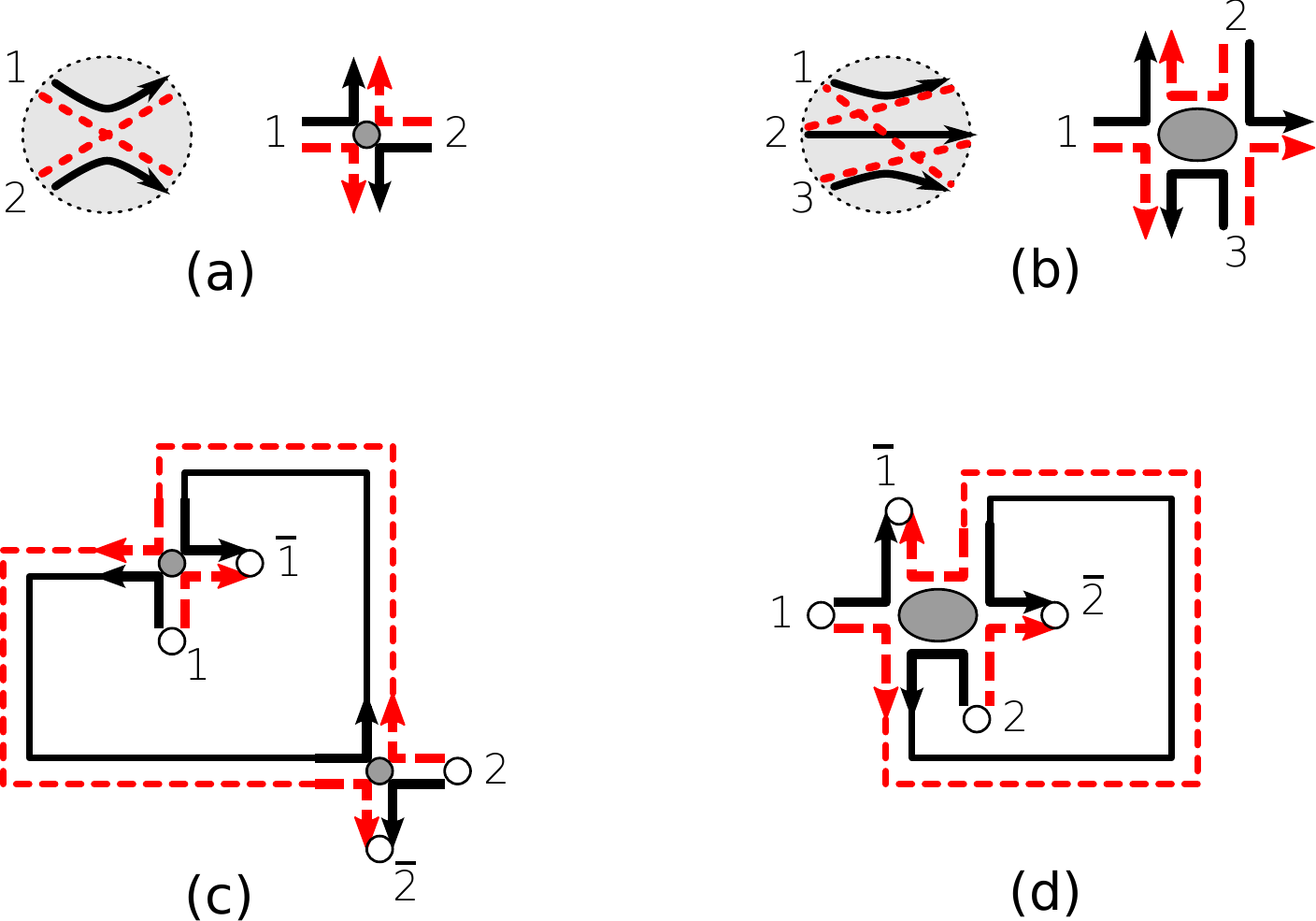}
  \caption{Untwisting the encounters into the vertices of the ribbon
    graph, (a) and (b).  The ribbon graphs (c) and (d) correspond to the
    diagrams of Fig.~\ref{fig:encounter_physics}.  To
    read off the trajectories, we start at the open end labelled $1$
    and follow the left side for $\gamma$ or the right side for
    $\gamma'$.  The leaves (vertices of degree 1) of the graph are
    shown as empty circles; the internal vertices are represented by
    the filled ellipses.  Edges going to leaves are normally drawn
    short to save space.  Other edges often have rectangular
    corners; the corners carry no particular meaning and were employed only
    due to the lack of artistic skill.}
  \label{fig:encounter_untwist}
\end{figure}

We ``untwist'' the encounter regions as shown in
Figs.~\ref{fig:encounter_untwist}(a) and (b), making sure the trajectories do
not intersect at encounters (vertices).  Some intersections or twists
might still be necessary, due to high genus or non-orientability of
the graph, but these are shifted away from vertices to avoid
overloading them with detail.  This way of drawing allows the
interpretation of the trajectories as parts of the boundary of the
resulting ribbon graphs and opens the area up to a variety of existing
combinatorial methods.

We formalize the above paragraphs as a definition.

\begin{definition}
  \label{def:unit_diagram}
  The \emph{unitary diagram} with the target permutation $\tau$ is a
  map satisfying the following:
  \begin{enumerate}
  \item There are $t$ vertices of degree 1 (henceforth \emph{leaves}) labelled
    with symbols $1,\ldots, t$ and $t$ leaves labelled
    with symbols $\wb{1},\ldots, \wb{t}$.
  \item All other vertices have even degree greater than $2$.
  \item \label{enum:boundary} A portion of the boundary running from
    one leaf to the next is called a \emph{boundary segment}.  Each
    leaf $z$ is incident to two boundary segments, one of which is a
    segment running between leaves $j$ and $\wb{j}$ and the other
    between leaves $\tau(k)$ to $\wb{k}$ (where either $z=j=\tau(k)$
    or $z=\wb{j}=\wb{k}$).  The segments are given direction $j \to
    \wb{j}$ and $\tau(k)\to\wb{k}$ and marked by solid and dashed
    lines correspondingly.  The following conditions are satisfied:
    \begin{enumerate}
    \item each part of the boundary is marked exactly once,
    \item \label{enum:two_sides} each edge is marked solid on one side
      and dashed on the other, both running in the same direction.
    \end{enumerate}
  \end{enumerate}
\end{definition}

We will now discuss some properties of the diagrams that follow from
the basic definition above.

The boundary segments marked solid correspond to the
$\gamma$-trajectories and dashed correspond to $\gamma'$-trajectories.
The boundary walks of the map can be read by alternatingly following
$\gamma$ and $\gamma'$ (in reverse) trajectories.  Then each edge is
traversed in opposite directions on the two sides, which means the
graph is orientable \cite[Chap. X]{Tutte_GraphTheory}. The labels of the
leaves are arranged in a special way: a trajectory $\gamma$ starts at
a vertex $j$ and ends at $\wb{j}$.  The trajectory $\gamma'$ that
immediately follows it, starts (when read in reverse) at $\wb{j}$ and
ends at $\tau(j)$, where $\tau$ is the target permutation of the
diagram in question, see the discussion of
Section~\ref{sec:target_unitary}.  To summarize,

\begin{lemma}
  \label{lem:unit_prop} 
  A unitary diagram satisfies the following properties:
  \begin{enumerate}
  \item \label{enum:orientable} The map is orientable.
  \item \label{enum:labelling} Each boundary walk passes through a
    non-zero even number of leaves.  Their labels form the sequence of
    the form 
    \begin{equation*}
      j,\ \wb{j},\ \tau(j),\ \wb{\tau(j)},\ \tau^2(j),\ \ldots,\
      \wb{\tau^f(j)},
    \end{equation*}
    where $\tau$ is the target permutation of the diagram and
    $\tau^{f+1}(j) = j$.  This establishes a one-to-one correspondence
    between the boundary walks and the cycles of $\tau$.
 \end{enumerate}
\end{lemma}

%%%%%%%%%%%%%%%%%%%%%%%%%%%%%%%%%%%%%%%%%%%%%%%%%%%%%
\subsection{Operations on diagrams}
\label{sec:operations_u}

We now define some operation on diagrams which will later allow us to find
cancellations of their contributions evaluated according to Definition~\ref{def:Essen_ansatz}.

\begin{figure}[t]
  \centering
  \includegraphics[scale=0.9]{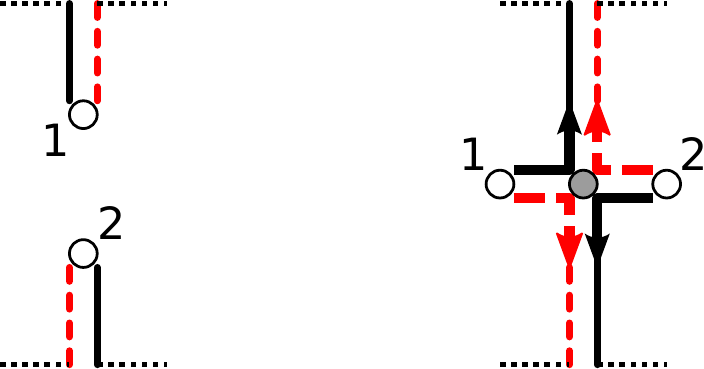}
  \caption{The operation of tying two leaves together.  By dotted
    lines we represent the rest of the graph.}
  \label{fig:tying}
\end{figure}

The first operation is \emph{tying} together two leaves.  Given two
leafs $j$ and $k$ (with no bars) we form a new vertex of degree 4 as shown in
Fig.~\ref{fig:tying}.  Note that the trajectory $\gamma$ connected to
the vertex $j$ is still connected to $j$ after tying, whereas the
$\gamma'$ trajectory becomes connected to $k$.  Tying two leaves
together creates an internal vertex and two edges.

The reversal of this operation is \emph{untying}.  It can only be
performed if there are two leaves (without bars) attached on the
opposite sides of the vertex of degree 4 that is to be untied.

\begin{lemma}
  \label{lem:tying_target}
  Consider a diagram with the target permutation $\tau$.  If we tie the
  leaves $j$ and $k$ together, the target permutation of the modified
  diagram is $(j\,k)\tau$.  Untying a 4-vertex with leaves $j$ and
  $k$ directly attached to its opposite sides also results in the
  target permutation $(j\,k)\tau$.
\end{lemma}

\begin{proof}
  The target permutation is read off a diagram by following a
  $\gamma$-trajectory (solid, black lines in the figures) from a leaf
  $i$ to the next leaf (which must be $\wb{i}$) and then following the
  $\gamma'$-trajectory (dashed, red lines) in reverse until the next
  leaf which is the leaf $\tau(i)$.  When the leaves $j$ and $k$ are
  tied together, the $\gamma$-trajectories are not changed, but the
  $\gamma'$-trajectories that originated from $j$ and $k$ are switched
  around.  Thus the target permutation of the modified diagram acts as
  the permutation $\tau$ followed by the interchange of $j$ and $k$.
  Hence the multiplication by $(j\,k)$ \emph{on the left}.

  Untying the vertex with leaves $j$ and $k$ is thus equivalent to
  multiplying by the inverse of $(j\,k)$, which is the same as
  multiplying by $(j\,k)$ itself.
\end{proof}

This operation of tying can be generalized to the case of two leaves
with bars and to the case of more than two leaves (of the same kind -- 
either all with or all without bars).  
The reversal, untying, will only work on a vertex of degree
$2m$ if there are $m$ leaves of the same kind attached to it.
In terms of the target permutations, the tying of several leaves
without bars is equivalent to multiplying by the cycle of length $m$.
The tying of leaves with bars is multiplication by the cycle \emph{on 
the right} and untying is multiplication by the inverse cycle on the
appropriate side.  The precise order of the elements in the cycle has
to agree with the spatial arrangement of the leaves around the cycle.
However, we will not need these operations in our proofs and thus 
omit the precise details here.  We will need these generalisations however 
to obtain moment generating functions as in the accompanying paper \cite{bk13b}, 
so we provide full details there.

The second operation we will need is \emph{contracting an edge}.  While any
edge connecting two (internal) vertices in a diagram can be contracted to
produce another diagram, we will only apply this operation in the
precise circumstances described below.

\begin{figure}[t]
  \centering
  \includegraphics[scale=0.9]{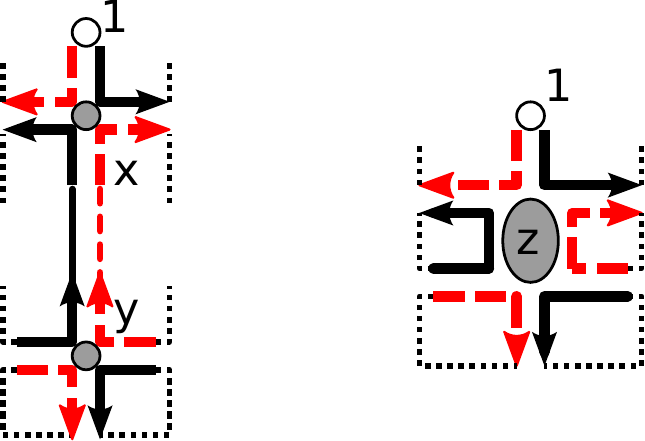}
  \caption{The operation of contracting an edge.  By dotted
    lines we represent the rest of the graph.}
  \label{fig:shrinking}
\end{figure}

In order to contract an edge going between vertices $x$ and $y$, see
Fig.~\ref{fig:shrinking}, we require that the vertex $x$ has degree 4
and that the edge attached to $x$ \emph{opposite} the edge $(x,y)$ is
coming directly from the leaf number $1$.  The vertex $y$ has no
restrictions, in particular it can be of any even degree.

The operation of contracting an edge leads to the diagram with a larger
vertex, shown on the right of Fig.~\ref{fig:shrinking}.  We can
reverse the operation if we have a vertex of an even degree larger
than 4 and an edge connecting the vertex to leaf $1$.  As the figure
suggests, the new vertex inherits the edge to $1$ and the two neighbouring
edges.  This operation is called \emph{splitting the vertex}.  The
following lemma is obvious.

\begin{lemma}
  \label{lem:shrinking_target}
  The operations of contracting an edge or splitting the vertex do not
  change the target permutation.
\end{lemma}

In some cases considered later the diagram will have no leaf with
number $1$.  In such a case the leaf with the smallest number will
play the role of leaf $1$.

%%%%%%%%%%%%%%%%%%%%%%%%%%%%%%%%%%%%%%%%%%%%%%%%
\subsection{Cancellations among unitary diagrams}
\label{sec:cancel_unit}

The result of Lemma~\ref{lem:structure_of_correlator} states that the
semiclassical correlator $C_E^{\U}(\boldsymbol{a},\boldsymbol{b})$ can be expressed as
\begin{equation*}
  C_E^{\U}(\boldsymbol{a},\boldsymbol{b}) = \delta_{t,s} \sum_{\sigma,\pi \in
    S_t} \Delta^{\U}(\tau) \prod_{j=1}^{t} \delta\!\left({a_j}-{b_{\sigma(j)}}\right) 
  \, \delta\!\left({a_\wb{j}}-{b_\wb{\pi(j)}}\right),
\end{equation*}
where $\Delta^{\U}(\tau)$ is the sum of contributions of diagrams with
the target permutation $\tau = \sigma^{-1}\pi$.  The contribution of
every diagram is evaluated according to the rules in
Definition~\ref{def:Essen_ansatz}.  Namely, each internal vertex of a
diagram gives a factor of $(-N)$ and each edge (including the ones
leading to the leaves) a factor of $1/N$.  Denoting by
$D_{\V,\E}^{\U}(\tau)$ the number of diagrams with the target
permutation $\tau$, $\V$ internal vertices and $\E$ edges, we have the
following statement, which finished the proof of
Theorem~\ref{thm:main_equivalence} in the case of broken time-reversal
symmetry.

\begin{theorem}
  \label{thm:sum_diag_u}
  The total contribution of the unitary diagrams with target permutation
  $\tau$ is
  \begin{equation}
    \label{eq:sum_diag_u}
    \Delta^{\U}(\tau) := \sum_{\V,\E} D_{\V,\E}^{\U}(\tau)
    \frac{(-1)^\V}{N^{\E-\V}} = V_N^{\U}(\tau). 
  \end{equation}
\end{theorem}

\begin{proof}
  The proof of the theorem has two parts.  In the first part we
  exhibit cancellations among the diagrams.  In the second part, the
  diagrams that survive cancellations will be shown to satisfy the
  same recursion as $V_N^{\U}(\tau)$ with the same initial conditions.

  The cancellation is based on the operations described in
  Section~\ref{sec:operations_u}.  We define an involution $P$ on the
  set of diagrams which changes $\V$ by one but leaves $(\E-\V)$
  invariant.  This means that the two diagrams that are images of each
  other under $P$ produce contributions to the sum in
  \eqref{eq:sum_diag_u} that are equal in magnitude but opposite in
  sign.  The sum therefore reduces to the sum over the fixed points of
  $P$.  Since each diagram has many edges and vertices, the main
  difficulty in describing $P$ is to uniquely specify the edge to be
  contracted or the vertex to be split.

  \begin{figure}[t]
    \centering
    \setlength{\unitlength}{\textwidth}
    \begin{picture}(0.95,0.3)%
      \put(0.0,0.03){\includegraphics[scale=0.6]{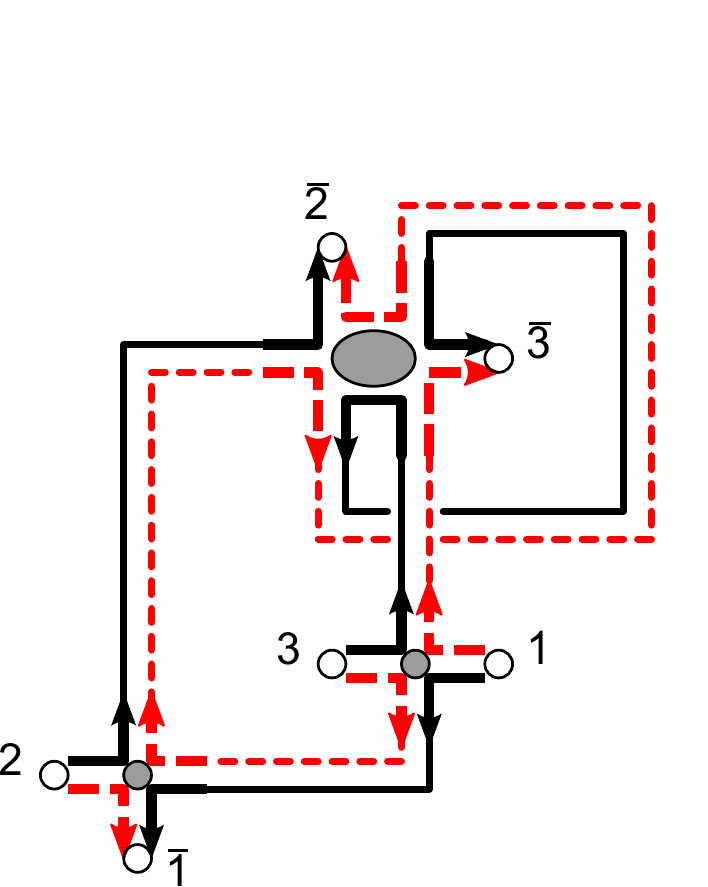}}
      \put(0.1,0){\makebox(0,0)[lb]{(a)}}%
      \put(0.25,0.03){\includegraphics[scale=0.6]{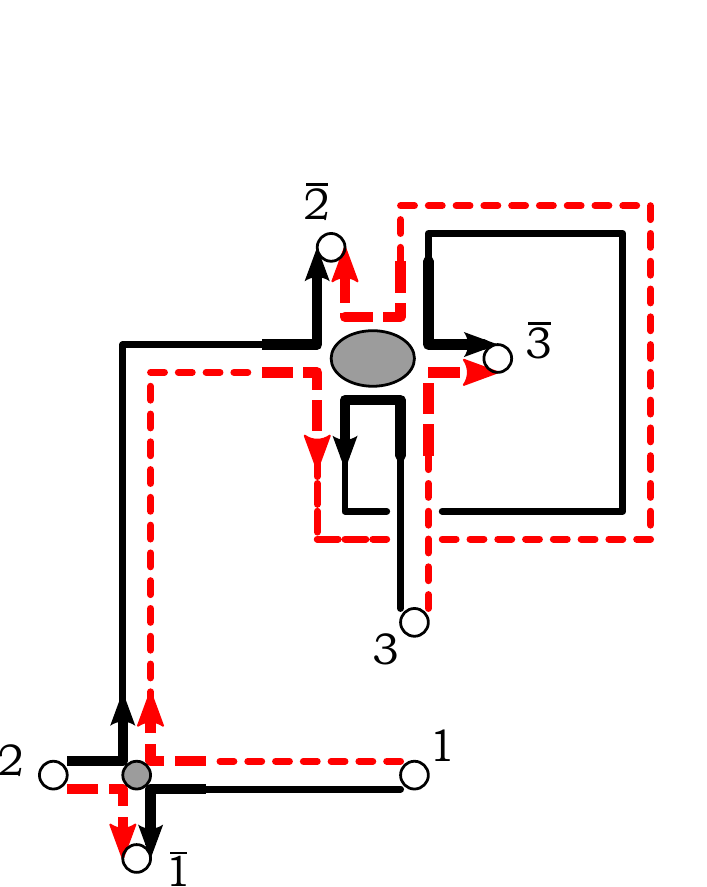}}
      \put(0.35,0){\makebox(0,0)[lb]{(b)}}%
      \put(0.52,0.055){\includegraphics[scale=0.6]{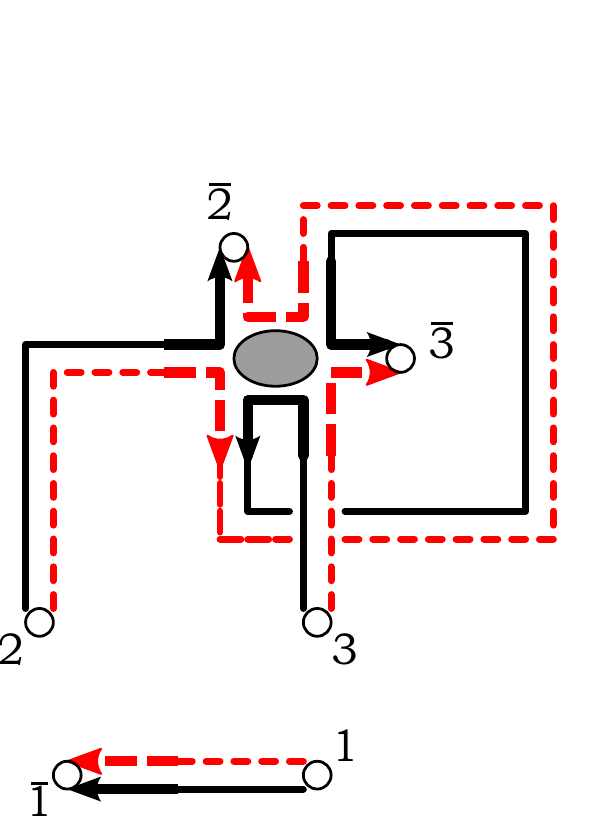}}
      \put(0.60,0){\makebox(0,0)[lb]{(c)}}%
      \put(0.75,0.11){\includegraphics[scale=0.6]{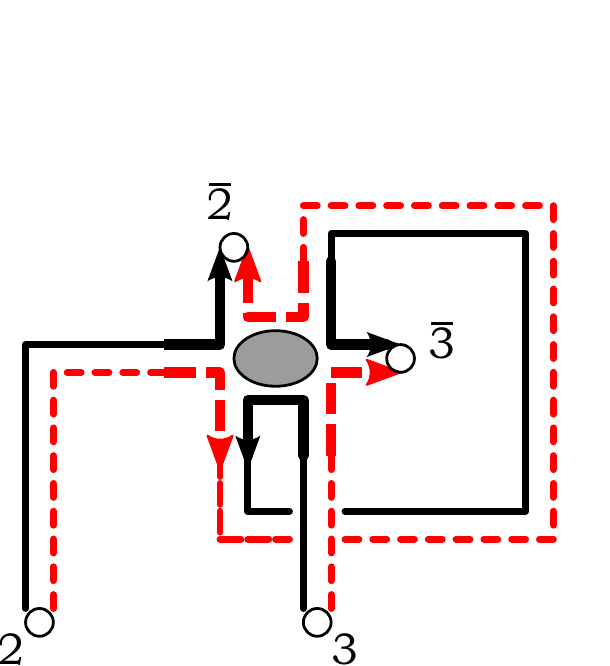}}
      \put(0.83,0){\makebox(0,0)[lb]{(d)}}%
    \end{picture}
    \caption{Repeatedly applying steps \ref{enum:minimal} and
      \ref{enum:untie} of the cancellation algorithm.  Part (a) is the
      original diagram.  The node $(1\,3)$ has been untied in (b); the
      node $(1\,2)$ has been untied in (c).  The edge connecting $1$ and
      $\wb{1}$ has been removed in (d) so that $2$ is the minimal leaf,
      but the diagram cannot be untied any further.}
    \label{fig:algorithmA}
  \end{figure}
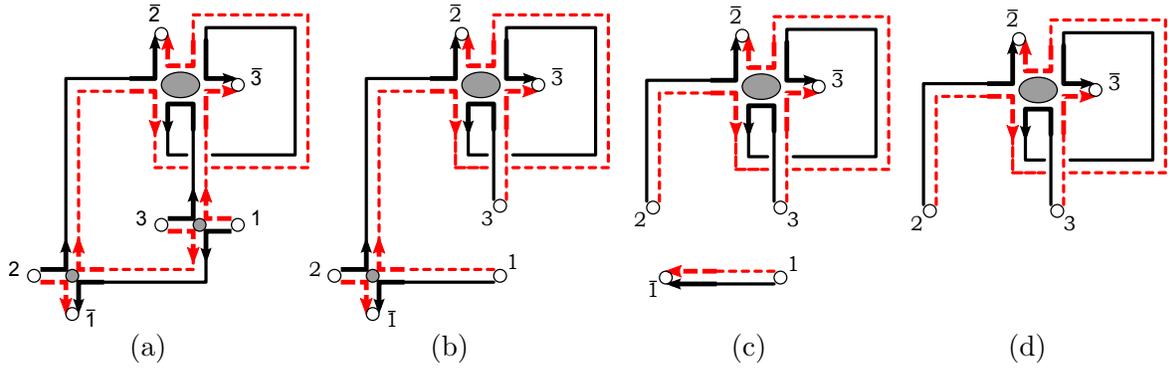

  The involution $P$ is described by the following algorithm:
  \begin{enumerate}
  \item \label{enum:minimal} Find the leaf with minimal number $j$ (without bar).
  \item \label{enum:untie} If it is attached to a vertex of degree 4
    whose opposite edge ends in a leaf, \emph{untie} the node and 
    return to step \ref{enum:minimal}.  If the untying produces an edge 
    which directly connects two leaves ($k$ and $\wb{k}$), this edge is 
    removed from consideration.
  \item \label{enum:shrink} If the leaf $j$ is attached to a node of
    degree 4 whose opposite edge leads to another internal vertex,
    \emph{contract} the edge.
  \item \label{enum:separate} Otherwise, the leaf $j$ is attached to a
    vertex of degree 6 or higher.  We \emph{split} this vertex into
    two by inserting a new edge.
  \item \label{enum:reverse} Reverse all operations performed in step~\ref{enum:untie}.
  \end{enumerate}
  
  The algorithm is illustrated by an example in
  Figs.~\ref{fig:algorithmA} and \ref{fig:algorithmB}.  The original
  diagram is Fig.~\ref{fig:algorithmA}(a) and the result is in
  Fig.~\ref{fig:algorithmB}(g).  However, if we start with
  Fig.~\ref{fig:algorithmB}(g) and apply the algorithm, we would use
  step \ref{enum:shrink} instead of step \ref{enum:separate} and
  arrive at Fig.~\ref{fig:algorithmA}(a).

  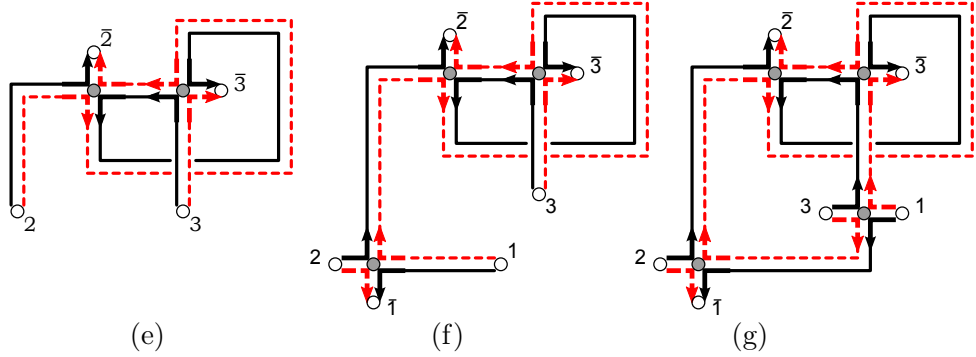
\begin{figure}[t]
    \centering
    \setlength{\unitlength}{\textwidth}
    \begin{picture}(0.75,0.3)%
      \put(0.0,0.1){\includegraphics[scale=0.6]{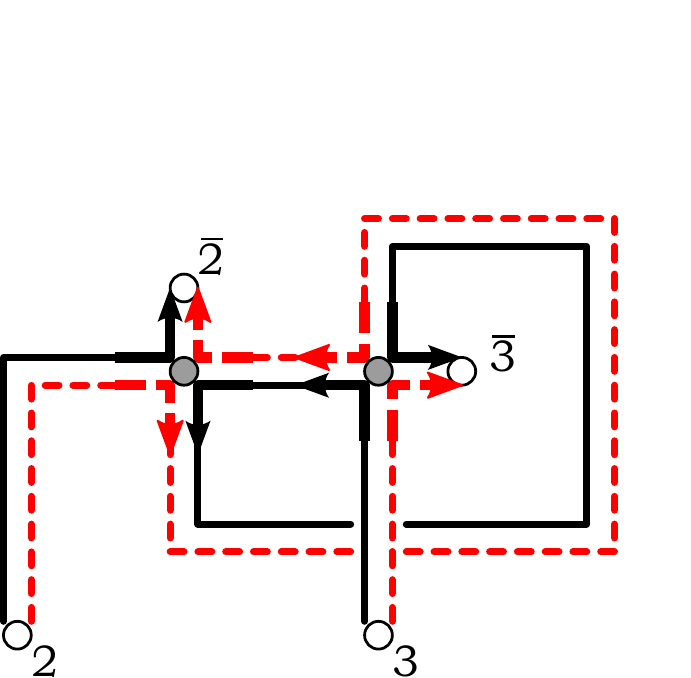}}
      \put(0.1,0){\makebox(0,0)[lb]{(e)}}%
      \put(0.25,0.03){\includegraphics[scale=0.6]{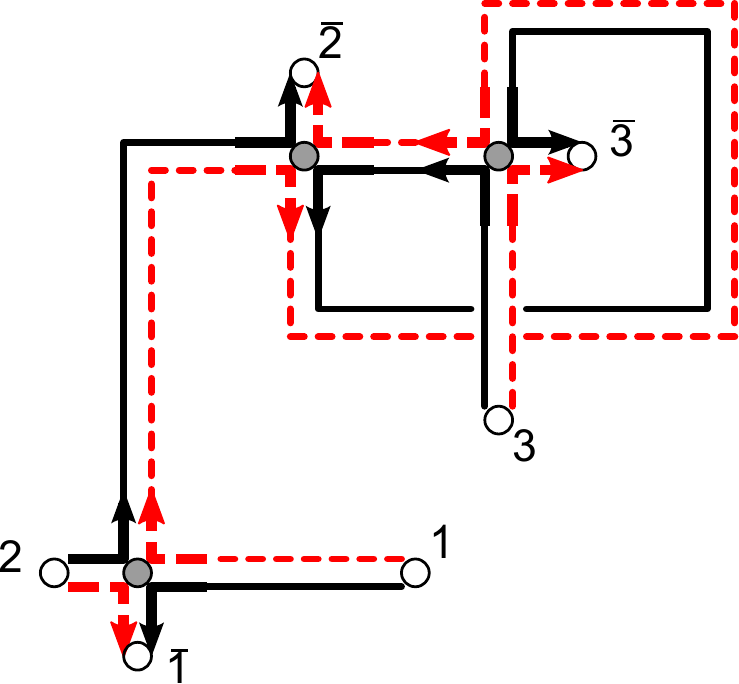}}
      \put(0.35,0){\makebox(0,0)[lb]{(f)}}%
      \put(0.52,0.03){\includegraphics[scale=0.6]{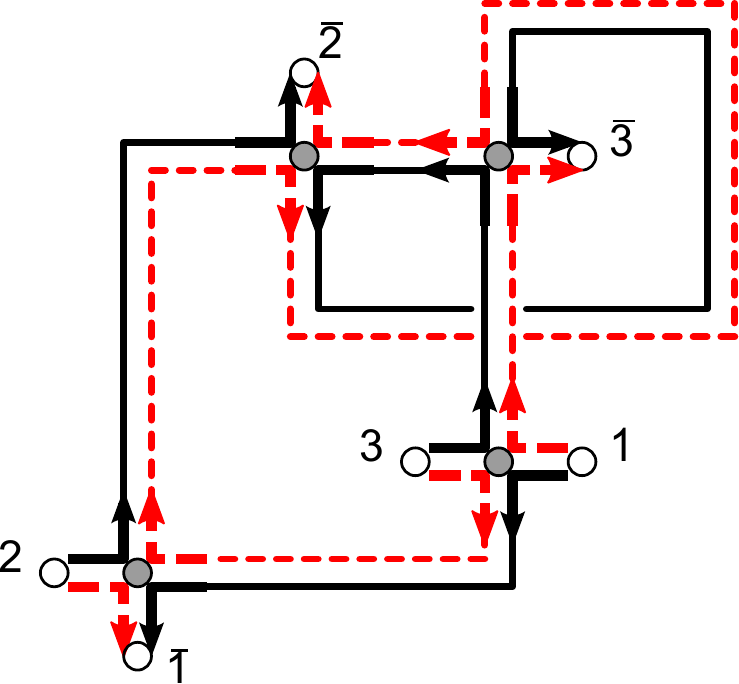}}
      \put(0.60,0){\makebox(0,0)[lb]{(g)}}%
   \end{picture}
    \caption{Applying steps \ref{enum:separate} and \ref{enum:reverse}
      of the cancellation algorithm.  Part (e) is the result of
      splitting the 6-vertex of Fig.~\ref{fig:algorithmA}(d) into
      two.  Parts (f) and (g) are tying of the leaves $(1\,2)$ and
      $(1\,3)$ correspondingly.  The diagram in part (g) is the
      counterpart of the original diagram in
      Fig.~\ref{fig:algorithmA}(a).}
    \label{fig:algorithmB}
  \end{figure}

  The algorithm preserves the target permutation (see
  Lemma~\ref{lem:shrinking_target}), therefore both diagrams
  contribute to the same sum.  Note that at most one of the operations
  in steps \ref{enum:shrink} and \ref{enum:separate} are performed.  
  These operations are the inverses of each other and provide
  the required difference in the number $\V$ while leaving $(\E-\V)$
  invariant.  If the steps \ref{enum:shrink} or \ref{enum:separate}
  are never reached, by virtue of the diagram untying
  completely, the corresponding contribution does not cancel and has
  to be counted.
  
  Consider a diagram that unties completely (for example, the diagram
  in Fig.~\ref{fig:NTRtraj_adv}(a)).  Recording the untyings,
  as suggested by Lemma~\ref{lem:tying_target}, we get 
  \begin{equation}
    \label{eq:untying}
    (s_\V\,r_\V) \cdots (s_2\,r_2)\,(s_1\,r_1) \, \tau = {id},
    \qquad s_j < r_j,
  \end{equation}
  where the identity permutation on the right corresponds to the
  resulting ``empty'' diagram.  The condition in
  step~\ref{enum:minimal} of the algorithm ensures that if $j < k$
  then $s_j \leq s_k$.  Given the sequence of untyings, we can
  reconstruct the diagram uniquely, therefore the diagrams that
  survive the cancellations are in one-to-one correspondence with the
  factorizations of $\tau$ into a product of transpositions satisfying
  \begin{equation}
    \label{eq:prim_fact}
    \tau = (s_1\,r_1) \, (s_2\,r_2) \cdots (s_\V\,r_\V),
    \qquad s_j < r_j, \quad s_j \leq s_{j+1}.
  \end{equation}
  These are known as \emph{primitive factorizations}\footnote{The
    traditional definitions are slightly different from ours in
    ordering the second elements of the transpositions} that are
  counted by \emph{monotone single Hurwitz numbers}
  \cite{GouGuaNov_prep11,MatNov_fpsac10}.  It is known
  \cite{MatNov_fpsac10} that the number of such factorizations
  provides the coefficients in the asymptotic expansion of the class
  coefficients $V_N^{\U}(\tau)$.  Here we provide a basic alternative
  proof of this fact that will be easy to generalize to the orthogonal
  case.
  
  Denote by $p_{\V}^{\U}(\tau)$ the number of primitive factorizations of the
  target permutation $\tau$ into $\V$ transpositions.  We have shown
  that
  \begin{equation}
    \label{eq:sum_diag_u_prim}
    \Delta^{\U}(\tau) = \sum_{\V} p_{\V}^{\U}(\tau)
    \frac{(-1)^\V}{N^{\E-\V}},
  \end{equation}
  where the number of edges $\E$ can be related to the other quantities
  as follows.  First, the permutation $\tau$ is a permutation on $t$
  elements.  The completely untied graph with $2t$ leaves contains $t$
  edges.  Each tying increases the number of edges by 2 and the number
  of internal vertices by one.  Since each of the $\V$ transpositions
  corresponds to a tying, we end up with $\E=t+2\V$ edges.

  The quantity $p_{\V}^{\U}(\tau)$ depends only on the cycle structure
  of $\tau$ so let $c_1,\ldots,c_k$ be the lengths of the cycles.
  Without loss of generality, we take $\tau=(1\,2\ldots
  c_1)\cdots(t-c_k+1\ldots t)$.  Consider the term $(s_1\, r_1)$ on
  the left of a primitive factorization of $\tau$.  Without this term,
  it is also a factorization, but of the permutation $(s_1\, r_1)\tau$
  and into $\V-1$ factors.

  For $s_1=1$, we will now investigate the cycle structure of the
  permutation $(1\, r_1)\tau$.  If $r_1$ belongs to the first cycle of
  $\tau$ (i.e.\ the one that contains 1) it splits into two, of
  lengths $q$ and $r$ with $q+r=c_1$.  Otherwise, if $r_1$ belongs to
  cycle number $j>1$, the first cycle joins with it to form a cycle of
  length $c_1+c_j$.

  Finally, it can happen that $s_1 \neq 1$, but only if the element
  $1$ does not appear in any transposition of the factorization
  (equivalently, the leaf $1$ is attached directly to the leaf
  $\wb{1}$).  This can happen only if $1$ is in cycle of its own
  within the permutation $\tau$, that is if $c_1=1$.  In this case,
  the given primitive factorization is also a factorization of the
  permutation on $t-1$ element with cycle lengths $c_2, \ldots, c_k$.
  Altogether we have
  \begin{equation}
    \label{eq:recur_prim}
    p_{\V}^{\U}(c_1,\ldots,c_k) = \delta_{c_1,1}p_{\V}^{\U}(c_2,\ldots,c_k)
    +\sum_{q+r=c_1}p_{\V-1}^{\U}(q,r,c_2,\ldots,c_k)
    +\sum_{j=2}^{k}c_jp_{\V-1}^{\U}(c_1+c_j,\ldots,\hat{c_j},\ldots).
  \end{equation}
  The notation $\hat{c_j}$ again means that $c_j$ is removed from the
  lengths of the cycles.  We would like to substitute this recursion
  into (\ref{eq:sum_diag_u_prim}) and we need to calculate the power
  of $N^{-1}$ at which the terms of (\ref{eq:recur_prim}) would enter
  the sum in $\Delta^U$.  On the left, we have $\E-\V = (t+2\V) - \V =
  t+\V$.  The first term on the right corresponds to permutations on
  $t-1$ elements, hence the power is $t+\V-1$.  The other two terms on
  the right have the same $t$, but a reduced number of internal
  vertices (factors in the factorization), also resulting in $t+\V-1$.
  Altogether, we get
  \begin{equation}
    N \Delta^{\U}(c_1,\ldots,c_k) = \delta_{c_1,1}\Delta^{\U}(c_2,\ldots,c_k)
    -\sum_{q+r=c_1}\Delta^{\U}(q,r,c_2,\ldots,c_k)
    -\sum_{j=2}^{k}c_j\Delta^{\U}(c_1+c_j,\ldots,\hat{c_j},\ldots),
  \end{equation}
  exactly mirroring the recursion relations for $V_N^{\U}$ given in
  (\ref{eq:recur_CUE}).
\end{proof}

\begin{example}
  \label{ex:contrib_V11}
  Consider the contribution $\Delta^{\U}(\id_2)$, where
  $\id_2$ is the identity permutation on 2 elements.  According to
  Lemma~\ref{lem:structure_of_correlator}, this contribution arises
  for the correlator $\langle S_{12} S_{34} S^*_{12}
  S^*_{34}\rangle$ or as part of many other correlators.
  The random matrix prediction for this is 
  \begin{equation*}
    V^{\U}(\id_2) = \frac1{N^2-1} = \frac{1}{N^2} + \frac{1}{N^4} + \ldots    
  \end{equation*}
  The first term arises from the trajectories $\gamma_1 = \gamma_1'$
  and $\gamma_2=\gamma_2'$, the corresponding diagram consisting of two
  disjoint edges connecting $1$ to $\wb{1}$ and $2$ to $\wb{2}$ correspondingly.

  \begin{figure}[t]
    \includegraphics[scale=0.7]{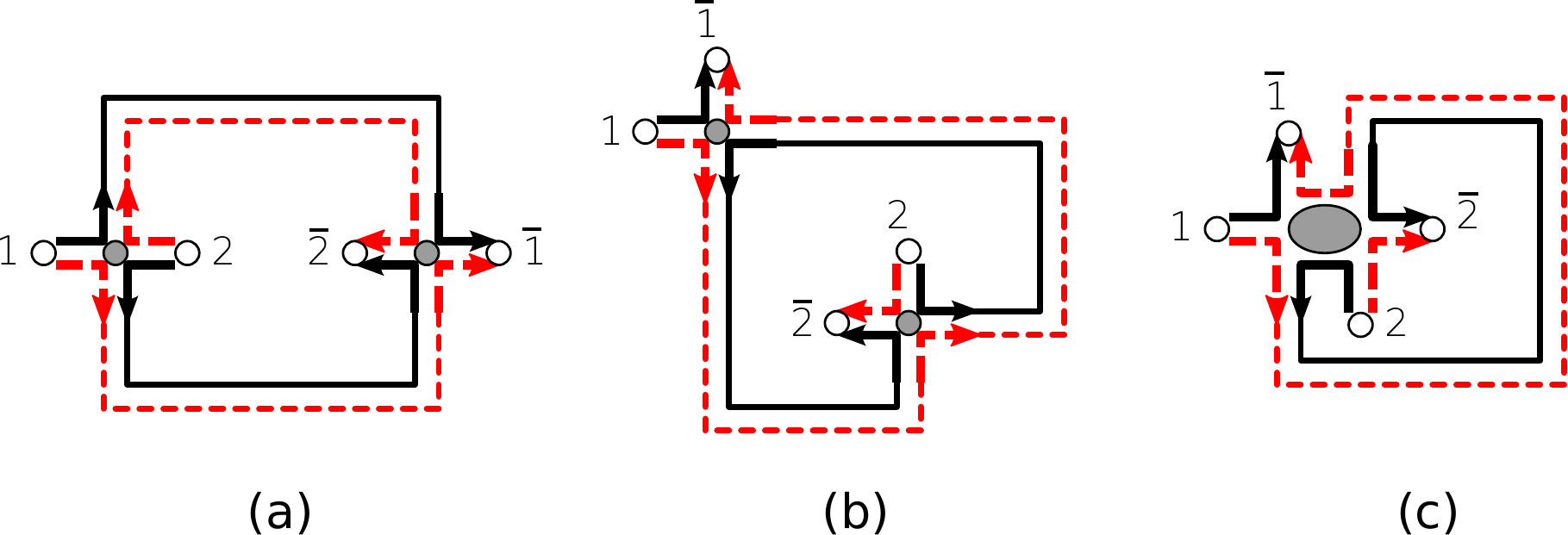}
    \caption{Diagrams contributing to the second term of the
      expansion of $V^{\U}(\id_2)$, see Example~\ref{ex:contrib_V11}.}
    \label{fig:example_contrib_V11}
  \end{figure}

  The diagrams contributing to the second term are depicted in
  Fig.~\ref{fig:example_contrib_V11}.  There are 5 diagrams
  altogether, with the additional two obtained from
  Fig.~\ref{fig:example_contrib_V11}(b) and (c) by a shift in the
  leaf labeling.  For example swapping the labels in 
  Fig.~\ref{fig:example_contrib_V11}(b) gives 
  Fig.~\ref{fig:encounter_physics}(a) and 
  Fig.~\ref{fig:encounter_untwist}(c).  Note that the diagrams in
  Fig.~\ref{fig:example_contrib_V11}(b) and (c) are related by the
  involution $P$ described in the proof of
  Theorem~\ref{thm:sum_diag_u}.  Indeed, in the diagram in part (b) we
  cannot untie any vertices.  We therefore go to step (3) of the
  algorithm and contract the edge going from the top shaded vertex to
  the right.  The result of this edge contraction is the diagram in
  part (c).  The diagram in part (b) has $v=2$ and $e=6$ and therefore
  contributes $(-1)^2/N^{(6-2)}$ to $\Delta^\U(\id_2)$.  The diagram in
  part (c) has $v=1$ and $e=5$ and contributes $-1/N^{5-1}$; these
  contributions cancel.

  The diagram in Fig.~\ref{fig:example_contrib_V11}(a) can be untied
  completely and corresponds to the primitive factorization
  \begin{equation*}
    \id = (1\,2)(1\,2).
  \end{equation*}
  In fact the only primitive factorizations are of the form 
  \begin{equation*}
    \id = (1\,2)^{2n},
  \end{equation*}
  recreating the random matrix prediction.
\end{example}

\begin{example}
  \label{ex:contrib_V2}
  Consider the contribution $\Delta^{\U}((1\,2))$, which is predicted by
  RMT to be
  \begin{equation*}
    \Delta^{\U}((1\,2)) = -\frac1{N(N^2-1)} = -\frac{1}{N^3} - \frac{1}{N^5} - \ldots    
  \end{equation*}

  \begin{figure}[t]
    \includegraphics[scale=0.7]{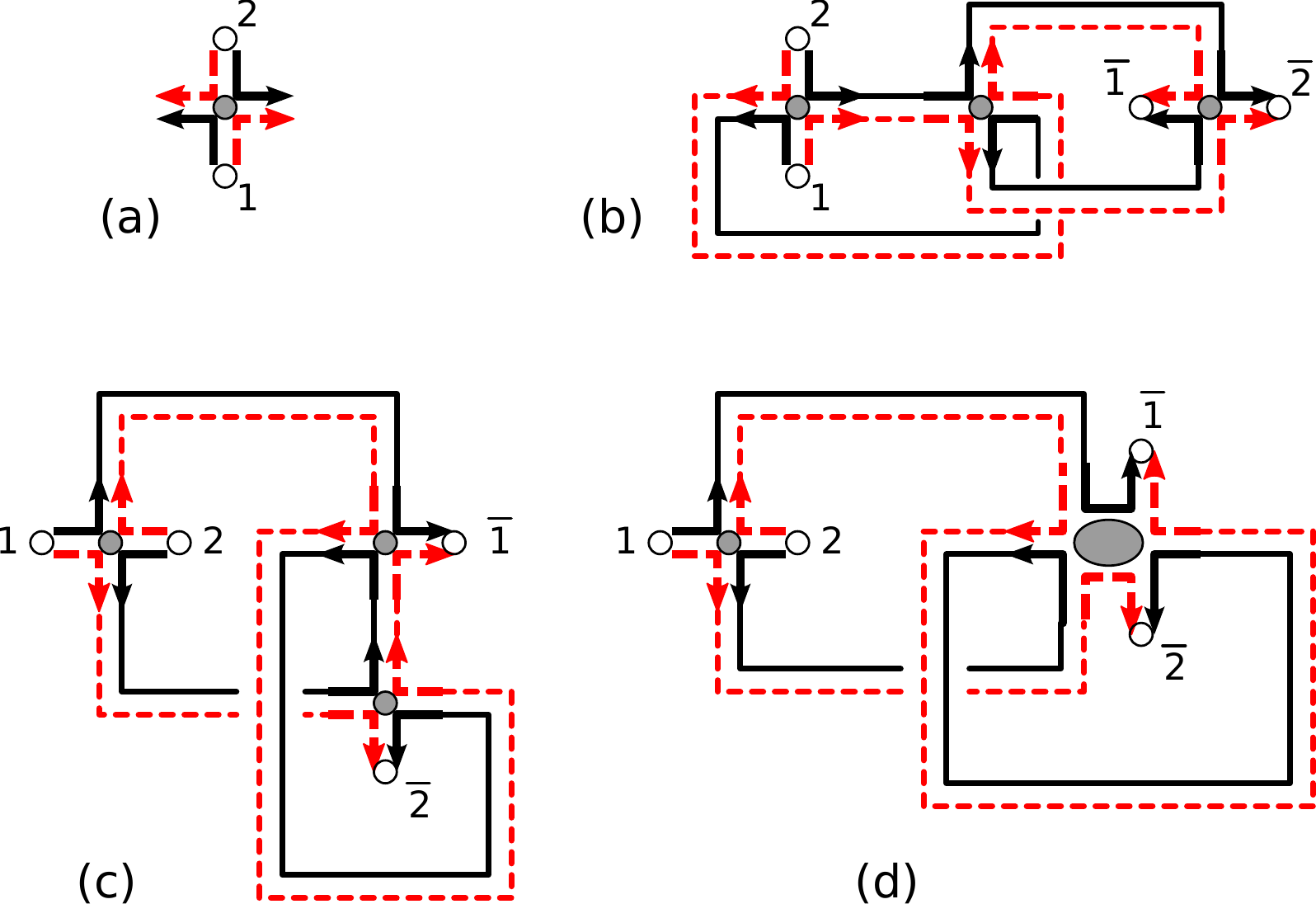}
    \caption{Diagrams contributing to the
      expansion of $V^{\U}((1\,2))$, see Example~\ref{ex:contrib_V2}.  }
    \label{fig:example_contrib_V2}
  \end{figure}

  The first term is produced by the diagram in
  Fig.~\ref{fig:example_contrib_V2}(a) with primitive factorization $(1\,2)$.
  The second term comes from the diagram in Fig.~\ref{fig:example_contrib_V2}(b),
  which corresponds to the factorization
  \begin{equation*}
    (1\,2)= (1\,2)(1\,2)(1\,2).
  \end{equation*}
  All other diagrams at this order (of which
  Figs.~\ref{fig:example_contrib_V2}(c) and (d) are two examples)
  mutually cancel.

  Considering the diagram in part (c) in more detail, we can untie the
  left shaded vertex.  The leaf number 1 is now attached to the top
  right shaded vertex, which cannot be untied.  We therefore contract
  the edge connected the two shaded vertices, tie leaves 1 and 2 back
  and arrive to the diagram in part (d).

  The higher order terms in the random matrix prediction come
  from the unique primitive factorization
  \begin{equation*}
    (1\,2) = (1\,2)^{2n+1}.
  \end{equation*}
\end{example}

While it is possible, in principle, to list all diagrams contributing
to a given moment of conductance fluctuations, a modification of the
method described in the follow-up paper \cite{bk13b} is better suited
to the task.  We stress, however, that our results guarantee that the
semiclassical evaluation will always agree with RMT.

\section{Summation over the orthogonal diagrams}
\label{sec:orthogonal}

Building on the results for systems with broken TRS, we now turn to the case 
of TRS.  We first develop the mathematical description of the ``orthogonal'' 
diagrams before summing their contributions.

%%%%%%%%%%%%%%%%%%%%%%%%%%%%%%%%%%%%%%%%%%%%%%%%%%
\subsection{Orthogonal diagrams as ribbon graphs}

As before, in the description of orthogonal diagarms, variables $j$ or
$k$ refer to the leaf label that does not have the bar
(correspondingly, $\wb{j}$ is a label that does have the bar), while
$z$ denotes a label either with or without the bar.

The conditions that make a valid orthogonal diagram are almost
identical to the unitary case.  The only significant difference is
that trajectories $\gamma$ and $\gamma'$ do not have to run in the
same direction.

\begin{definition}
  \label{def:orth_diagram}
  The \emph{orthogonal diagram} with the target permutation $\tau$ is a
  locally orientable map satisfying the following:
  \begin{enumerate}
  \item There are $t$ leaves labelled with symbols $1,\ldots, t$ and
    $t$ leaves labelled with symbols $\wb{1},\ldots, \wb{t}$.
  \item All other vertices have even degree greater than $2$.
  \item Each leaf is incident to two boundary segments, one of which
    runs between labels $z_1$ and $\wb{z_1}$ and is marked solid, and
    the other runs between labels $\tau(z_2)$ and $\wb{z_2}$ and is
    marked dashed.  Each edge is marked solid on one side and dashed
    on the other.
  \end{enumerate}
\end{definition}

Since $\tau$ does not preserve the two ``halves'' of the set
$Z_t=\{1,\ldots,t, \wb{t}, \ldots, \wb{1}\}$, there can be boundary
segments running between $j$ and $k$ or between $\wb{j}$ and
$\wb{k}$ (see Fig.~\ref{fig:TRtraj_adv} for some examples).  Thus
there is no natural way to assign direction to the boundary segments.
However, when we consider moment generating functions \cite{bk13b} we will 
again have a natural direction and for this reason we
retain the directional arrows in the figures of orthogonal diagrams.

The consequence of dropping the direction requirement in
Definition~\ref{def:orth_diagram} is that we now
need both orientable and (globally) non-orientable diagrams.
Therefore, when drawn on a plane, some edges of the map might have
twists in them: going around a graph on a closed walk can bring you
back on the reverse side of the edge, see Figure~\ref{fig:TRtraj_adv}(b)
for an example.  The last property of Lemma~\ref{lem:unit_prop} still
applies to orthogonal diagrams.  In fact, the face labels fit the pattern
\begin{equation*}
  z,\ \wb{z},\ \tau(z),\ \wb{\tau(z)},\ \tau^2(z),\ \ldots,\
  \wb{\tau^f(z)},
\end{equation*}
\emph{independently} of the direction chosen for the boundary.  This
is because $\tau\left( \wb{\tau(z)} \right) = \wb{z}$ due to
Lemma~\ref{lem:orth_target_prop}.

Given an orthogonal diagram, we can read off the target permutation in
the following fashion.  To determine $\tau(z)$, we find the leaf
labelled $z$ and go from it along the boundary segment marked solid,
until we arrive at the next leaf (which must be marked $\wb{z}$,
according to Def.~\ref{def:unit_diagram}).  From there we follow the
dashed boundary segment.  The leaf we arrive at is the leaf $\tau(z)$.
Note that the same prescription applies to unitary diagrams as well,
if we let $z$ be the labels $\{1,\ldots,t\}$ only.

%%%%%%%%%%%%%%%%%%%%%%%%%%%%%%%%%%%%%%%%%%%%%%%%%%%%%
\subsection{Operations on diagrams}
\label{sec:operations_o}

The two types of operations for unitary diagrams described in Section~\ref{sec:operations_u} 
can be defined for orthogonal diagrams in a similar fashion.

In fact, the operation of contracting an edge is the exactly same.
Namely, if we have a vertex of degree 4 with the leaf labelled $1$
adjacent to it, we can contract the edge that is opposite the leaf $1$.

The generalization of the tying/untying operation is a little more
exciting.  We can tie \emph{any two} leaves together (that is with or
without bars).  Before tying, the edges need to be arranged as in
Fig.~\ref{fig:tying}; this might require adding a twist to one of
them, see Fig.~\ref{fig:tying_TR}.  

\begin{figure}[t]
  \centering
  \includegraphics[scale=0.6]{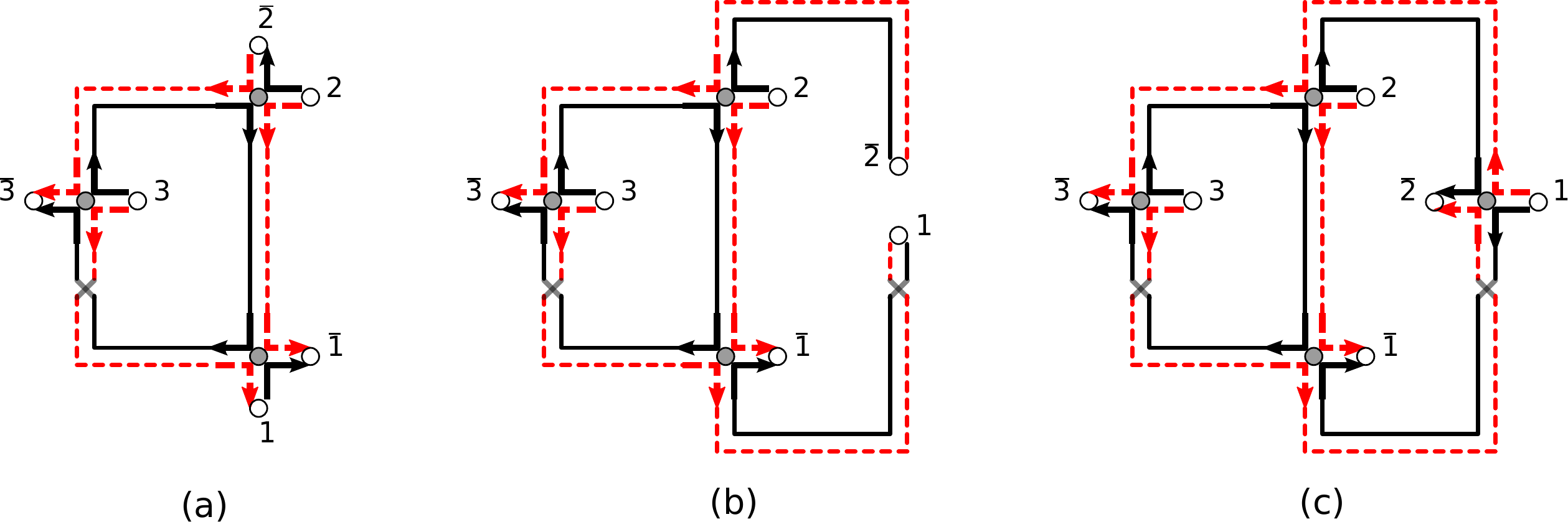}
  \caption{An example of tying two leaves in an orthogonal diagram.
    Part (a) shows the original diagram; leaves $1$ and $\wb{2}$ are to
    be tied.  In part (b) one of the edges received a twist to align
    properly.  The result is shown in part (c).  According to
    Lemma~\ref{lem:tying_target_O}, the target is transformed as
    $(1\,\wb2)(1\,2\,\wb3)(3\,\wb2\,\wb1)(2\,\wb1) = (1\,2\,3)(\wb3\,\wb2\,\wb1)$}
  \label{fig:tying_TR}
\end{figure}

\begin{lemma}
  \label{lem:tying_target_O}
  Consider a diagram with the target permutation $\tau$.  If we tie the
  leaves $z_1$ and $z_2$ together, the target permutation of the
  modified diagram is $(z_1\,z_2)\,\tau\,(\wb{z_1}\,\wb{z_2})$.
  Untying a 4-vertex with leaves $z_1$ and $z_2$ directly attached to
  its opposite sides also results in the target permutation
  $(z_1\,z_2)\,\tau\,(\wb{z_1}\,\wb{z_2})$.
\end{lemma}

\begin{proof}
  According to the rules of reading off the target permutation, the
  result of $\tau$ is the endpoint of a dashed boundary segment.
  After tying the leaves, the dashed segment that was finishing at
  $z_1$ now finishes at $z_2$ (see Figure~\ref{fig:tying}) and vice
  versa.  This means that $z_1$ and $z_2$ must be switched after the
  original $\tau$ is applied.

  Similarly, the image of $\wb{z_1}$ is computed by following the
  solid segment to $z_1$ and then the dashed segment to the leaf that
  is the result $\tau\left(\wb{z_1}\right)$.  After the tying
  operation, the dashed segment coming out of $z_1$ is actually the
  segment that was previously coming out of $z_2$ and therefore
  leading to $\tau\left(\wb{z_2}\right)$.  To account for this change,
  we need to switch $\wb{z_1}$ and $\wb{z_2}$ \emph{before} we apply
  $\tau$.

  The untying operation is equivalent to multiplying by the inverses
  of the transpositions, which are the transpositions themselves.
\end{proof}

\begin{remark}
  The operation $(z_1\,z_2)\,\tau\,(\wb{z_1}\,\wb{z_2}) =: \wt{\tau}$
  preserves the properties described in
  Lemma~\ref{lem:orth_target_prop}.  Indeed, according to
  Remark~\ref{rem:Ttau_involution}, we need to show that $\Tt\wt{\tau}$
  consists only of cycles of length 2.  Denoting $(\wb{z_1}\,\wb{z_2})
  =: q$, we have $(z_1\,z_2) = \Tt q\Tt$ and therefore $\Tt\wt{\tau} =
  \Tt\Tt q\Tt\tau q = q \Tt\tau q \sim \Tt\tau$.  Since the operation of
  conjugation does not affect the cycle type, $\Tt\wt{\tau}$ has the
  same cycle lengths as $\Tt\tau$.
\end{remark}

%%%%%%%%%%%%%%%%%%%%%%%%%%%%%%%%%%%%%%%%%%%%%%%%%
\subsection{Cancellations among the orthogonal diagrams}
\label{sec:cancel_orth}

Similarly to the unitary case, comparing
Lemma~\ref{lem:structure_of_correlator_O} with the RMT result
\eqref{eq:avO}, we see that we need to compare the coefficients
$\Delta^{\OO}(\tau)$ and $V^{\OO}_N(\varpi)$ (where $\tau=\varpi^{-1} \Tt
\varpi \Tt$).  Moreover, according to
Lemma~\ref{lem:even_odd_and_orth_target}, both coefficients in fact
depend on the set of numbers $c_1, \ldots, c_k$, which are the lengths
of half the cycles in the cycle representation of $\tau$ (each cycle has 
its ``mirror'' image in the representation, see Lemma~\ref{lem:orth_target_prop}; 
only one length per pair appears in the list).

The coefficient $\Delta^{\OO}(\tau)$ is the sum of contributions of
orthogonal diagrams with the target permutation $\tau$.  The
contribution of every diagram is evaluated according to the rules in
Definition~\ref{def:Essen_ansatz}. Denoting by $D^{\OO}_{\V,\E}(\tau)$ the
number of orthogonal diagrams with the target permutation $\tau$, $\V$
internal vertices and $\E$ edges, we have the following statement.

\begin{theorem}
  \label{thm:sum_diag_O}
  The total contribution of the orthogonal diagrams with target
  permutation $\tau$ is
  \begin{equation}
    \label{eq:sum_diag_o}
    \Delta^{\OO}(\tau) := \sum_{\V,\E} D^{\OO}_{\V,\E}(\tau)
    \frac{(-1)^\V}{N^{\E-\V}} = V_N^{\OO}(\tau). 
  \end{equation}
\end{theorem}

\begin{proof}
  The cancellation algorithm is completely analogous to the unitary
  case.  Namely, we fix a linear ordering of the set $Z_t$ and
  untie, while it is possible, the vertex adjacent to the leaf with
  the minimal label of the leaves still present in the diagram.  We
  remove any edges that directly connect two leaves from the consideration.

  If at any point it becomes impossible to untie the vertex adjacent
  to the minimal leaf, it is either because the vertex is of a degree
  higher than 4 or because the edge opposite the minimal leaf is not
  going to a leaf.  Then we correspondingly split the vertex or
  contract the edge.  After this we re-tie previously untied vertices.
  The new diagram has the contribution that cancels the contribution
  of the original diagram.

  The only diagrams that survive this process are those that untie to
  an empty diagram.  Recording every step according to
  Lemma~\ref{lem:tying_target_O}, we get
  \begin{equation}
    \label{eq:untied_O}
    (s_\V\,r_\V) \cdots (s_2\,r_2)\,(s_1\,r_1) \, \tau 
    \, (\wb{s_1}\,\wb{r_1}) \, (\wb{s_2}\,\wb{r_2}) 
    \cdots (\wb{s_\V}\,\wb{r_\V}) = id,
  \end{equation}
  where $s_j < r_j$ and $s_j \leq s_{j+1}$.  Note that the choice of
  the node to untie is unique at each step, so the surviving diagrams
  are in one-to-one correspondence with the ``palindromic'' primitive
  factorizations
  \begin{equation}
    \label{eq:prim_fact_O}
    \tau = (s_1\,r_1) \, (s_2\,r_2) \cdots (s_\V\,r_\V)
    (\wb{s_\V}\,\wb{r_\V}) \cdots (\wb{s_2}\,\wb{r_2}) \,
    (\wb{s_1}\,\wb{r_1}), 
    \qquad s_j < r_j, \quad s_j \leq s_{j+1}.
  \end{equation}
  We now need to understand the number $p^{\OO}_\V(\tau) =
  p^{\OO}_\V(c_1,\ldots, c_k)$ of such factorizations which provide us 
  with the semiclassical contribution.
  \begin{equation}
    \Delta^{\OO}(\tau) := \sum_{\V} p^{\OO}_{\V}(\tau)
    \frac{(-1)^\V}{N^{\E-\V}} . 
  \end{equation}
  To do this, we derive a recursion similar to \eqref{eq:recur_prim}.  The
  number of factorizations of $\tau$ with $2\V$ factors is equal to
  the number of factorization of $(s_1\, r_1)\, \tau \,
  (\wb{s_1}\,\wb{r_1})$ with $2(\V-1)$ factors, summed over all
  possible choices of $(s_1\, r_1)$.  As before, we treat the case $s_1 \neq
  \min(Z_t)$ separately ($s_1$ and $\wb{s_1}$ must then be cycles
  of their own).  
  
  Without loss of generality we assume that $s_1$ appears in cycle
  number 1.  There are $2c_j$ possibilities for $r_1$ to appear in the
  cycle number $j>1$ or its ``mirror'' image, cycle number $2k-j+1$.
  This cycle joins the first cycle in the result of the multiplication
  $(s_1\, r_1)\, \tau \, (\wb{s_1}\,\wb{r_1})$.  If $r_1$ belongs to
  the first cycle, it splits into two parts.  Finally, if $r_1$
  belongs to the mirror image of the first cycle, the product is
  \begin{equation*}
    (s_1\, r_1) \, (s_1\, a \ldots b\, \wb{r_1}\, c \ldots d)
    (\wb{d} \ldots \wb{c}\, r_1\, \wb{b} \ldots \wb{a}\, \wb{s_1})
    \, (\wb{r_1}\, \wb{s_1})
    = (s_1\, a \ldots b\, \wb{r_1}\,\wb{d} \ldots \wb{c}) (c \ldots
    d\, r_1\, \wb{b} \ldots \wb{a}\, \wb{s_1}),
  \end{equation*}
  therefore the cycle lengths do not change.  In the latter case,
  $r_1$ has $c_1$ possibilities, including $\wb{s_1}$.

  The cases above give rise to the following terms,
  \begin{multline}
    \label{eq:recur_orth}
    p^{\OO}_\V(c_1,\ldots,c_k) = \delta_{c_1,1} p^{\OO}_\V(c_2,\ldots,c_k)
    + \sum_{j} 2c_j p^{\OO}_{\V-1}(c_1+c_j,\ldots,\hat{c_j},\ldots) \\
    + \sum_{q+r=c_1}p^{\OO}_{\V-1}(q,r,c_2,\ldots,c_k)
    + c_1 p^{\OO}_{\V-1}(c_1,\ldots,c_k).
  \end{multline}
  Taking the generating function with respect to $1/N$ according to
  the middle expression in~(\ref{eq:sum_diag_o}), we recover
  recursion~(\ref{eq:recur_COE}), which completes the proof.
\end{proof}

\section{Conclusions}
\label{sec:conclusions}

While the RMT approach described in Sec.~\ref{sec:RMT} in terms of the
recursive class coefficients is computationally inefficient, it turned
out to be useful for establishing the equivalence between semiclassics
and RMT to all orders of $1/N$.  The equivalence is established at the
level of moments of scattering matrix elements which immediately
implies equivalence of moments (both linear and non-linear) of matrix
subblocks.  The result is proved for moments of COE and CUE
(corresp.\ with and without time-reversal symmetry), but, since there
is a simple formula \cite{bb96} connecting COE and CSE moments, the equivalence
extends immediately to CSE moments (corresp.\ systems with spin-orbit
interactions) as well.

Furthermore, since all the moments agree, indirectly we obtain the
joint probability density of the transmission eigenvalues of
$\bs{t}^\dagger \bs{t}$ semiclassically.  In fact, current RMT
approaches start from this joint probability density (which follows
the Jacobi ensemble \cite{beenakker97,forrester06}) and obtain linear
and non-linear transport moments through a variety of different
methods
\cite{kss09,lv11,ms11,ms13,novaes08,ok08,ok09,ss06,ssw08,vv08}.

Once the semiclassical diagrams have been expressed in terms of ribbon
graphs, the most important step is showing that the contributions of
the vast majority of diagrams cancel. We identified pairs of diagrams
whose contributions cancel exactly, leaving only the diagrams that
correspond to primitive factorisations. These were then shown to match
the RMT class coefficients. While possibly the simplest cancellation
algorithm, it is certainly not unique. In fact, in an earlier version
of our proof we considered a cancellation algorithm which reduced the
set of diagrams (with broken TRS) to those that can be put into
correspondence with inequivalent factorisations \cite{BerIrv_prep}
(those are the fully untieable diagrams, when more general untyings
are allowed).

The cancellation we used relies strongly on the product of the
semiclassical edge and vertex contribution from
Definition~\ref{def:Essen_ansatz} being exactly $-1$. This is the case
for the correlators of the subblocks of the scattering matrix we
considered here, but no longer holds for energy-dependent correlators,
superconducting or tunneling leads, Wigner delay times and other
physically relevant questions. In these cases semiclassical results
are known for low moments or up to a given order in inverse channel
number \cite{bk10,bk11,kuipersetal10,kuipersetal11,kr13} and agree
with perturbative RMT expansions
\cite{bmb95,bb96,bfb99,melsenetal96,melsenetal97,ms12,sss01} but a
general proof of the equivalence is lacking. For the Wigner delay
times, the transport matrix follows an inverse Wishart distribution
\cite{bfb99} and correlators of such matrices have been expressed
\cite{matsumoto12} in a form similar to the correlators in
Sec.~\ref{sec:RMT}. A strategy similar to the one used here might then
be fruitful. Starting from the joint probability density of the delay
times, the linear moments, as well as the moments of the mean time
delay, are also available from RMT \cite{lv11,ms11,ms13}.

Another case where the edge and vertex contributions do not allow
direct cancellation is when we consider Ehrenfest time effects. Below
the Ehrenfest time the quantum and classical propagation are fairly
similar while for longer times wave interference dominates. This wave
interference is incorporated in the semiclassical approximation on top
of the underlying classical motion which can then be used to obtain
the typical dependence of quantum transport on the Ehrenfest time. In
particular results are known for low orders in a perturbative
expansion in the inverse channel number for low moments
\cite{br06,br06b,jw06,petitjeanetal09,wk10,waltneretal12,whitney07,wj06}
and for all moments at leading order \cite{wkr11}. Although this
question is outside the range of applicability of RMT, attempts have
been made to phenomenologically treat Ehrenfest time effects using
``effective'' RMT \cite{sgb03,sgb03b}.  This approach provides the correct results for
some quantities but importantly not for all (notably the weak
localisation correction to the conductance) and its validity must then
be checked against semiclassical approaches \cite{br06b,jw06,wkr11}.  It is here, where RMT
answers are not available in principle, that our classification of the
semiclassical diagrams (continued in \cite{bk13b}) can become very
useful.

\section*{Acknowledgements}

Illuminating discussions with J.~Irving and K.~Richter are gratefully
acknowledged.  GB is partially supported by the NSF grant DMS-0907968 
and JK by the DFG through FOR 760.

%%% Local Variables: 
%%% mode: latex
%%% TeX-master: "../ctsetm1"
%%% End: 

\bibliographystyle{abbrv}
\bibliography{ctsetm}

\begin{thebibliography}{10}

\bibitem{bjs93a}
H.~U. Baranger, R.~A. Jalabert, and A.~D. Stone.
\newblock Weak localization and integrability in ballistic cavities.
\newblock {\em Phys. Rev. Lett.}, 70:3876--3879, 1993.

\bibitem{beenakker97}
C.~W.~J. Beenakker.
\newblock Random-matrix theory of quantum transport.
\newblock {\em Rev. Mod. Phys.}, 69:731--808, 1997.

\bibitem{bmb95}
C.~W.~J. Beenakker, J.~A. Melsen, and P.~W. Brouwer.
\newblock Giant backscattering peak in angle-resolved {A}ndreev reflection.
\newblock {\em Phys. Rev. B}, 51:13883--13886, 1995.

\bibitem{bhn08}
G.~Berkolaiko, J.~M. Harrison, and M.~Novaes.
\newblock Full counting statistics of chaotic cavities from classical action
  correlations.
\newblock {\em J. Phys. A}, 41:365102, 2008.

\bibitem{BerIrv_prep}
G.~Berkolaiko and J.~Irving.
\newblock Inequivalent transitive factorizations of permutations into
  transpositions.
\newblock In preparation, 2013.

\bibitem{bk10}
G.~Berkolaiko and J.~Kuipers.
\newblock Moments of the {W}igner delay times.
\newblock {\em J. Phys. A}, 43:035101, 2010.

\bibitem{bk11}
G.~Berkolaiko and J.~Kuipers.
\newblock Transport moments beyond the leading order.
\newblock {\em New J. Phys}, 13:063020, 2011.

\bibitem{bk12}
G.~Berkolaiko and J.~Kuipers.
\newblock Universality in chaotic quantum transport: The concordance between
  random matrix and semiclassical theories.
\newblock {\em Phys. Rev. E}, 85:045201, 2012.

\bibitem{bk13b}
G.~Berkolaiko and J.~Kuipers.
\newblock Combinatorial theory of the semiclassical evaluation of transport
  moments {II}.
\newblock preprint {\tt arXiv:1307.3280}, 2013.

\bibitem{BerSchWhi_jpa03}
G.~Berkolaiko, H.~Schanz, and R.~S. Whitney.
\newblock Form factor for a family of quantum graphs: an expansion to third
  order.
\newblock {\em Journal of Physics A: Mathematical and General}, 36(31):8373,
  2003.

\bibitem{berry85}
M.~V. Berry.
\newblock Semiclassical theory of spectral rigidity.
\newblock {\em Proc. Roy. Soc. A}, 400:229--251, 1985.

\bibitem{BluSmi_prl88}
R.~Bl\"umel and U.~Smilansky.
\newblock Classical irregular scattering and its quantum-mechanical
  implications.
\newblock {\em Phys. Rev. Lett.}, 60:477--480, 1988.

\bibitem{BluSmi_prl90}
R.~Bl\"umel and U.~Smilansky.
\newblock Random-matrix description of chaotic scattering: Semiclassical
  approach.
\newblock {\em Phys. Rev. Lett.}, 64:241--244, 1990.

\bibitem{braunetal06}
P.~Braun, S.~Heusler, S.~M\"uller, and F.~Haake.
\newblock Semiclassical prediction for shot noise in chaotic cavities.
\newblock {\em J. Phys. A}, 39:L159--L165, 2006.

\bibitem{bb96}
P.~W. Brouwer and C.~W.~J. Beenakker.
\newblock Diagrammatic method of integration over the unitary group, with
  applications to quantum transport in mesoscopic systems.
\newblock {\em J. Math. Phys.}, 37:4904--4934, 1996.

\bibitem{bfb99}
P.~W. Brouwer, K.~M. Frahm, and C.~W.~J. Beenakker.
\newblock Distribution of the quantum mechanical time-delay matrix for a
  chaotic cavity.
\newblock {\em Waves in Random Media}, 9:91--104, 1999.

\bibitem{br06}
P.~W. Brouwer and S.~Rahav.
\newblock Semiclassical theory of the {E}hrenfest time dependence of quantum
  transport in ballistic quantum dots.
\newblock {\em Phys. Rev. B}, 74:075322, 2006.

\bibitem{br06b}
P.~W. Brouwer and S.~Rahav.
\newblock Towards a semiclassical justification of the effective random matrix
  theory for transport through ballistic quantum dots.
\newblock {\em Phys. Rev. B}, 74:085313, 2006.

\bibitem{Buttiker86}
M.~B\"uttiker.
\newblock Four-terminal phase-coherent conductance.
\newblock {\em Phys. Rev. Lett.}, 57:1761--1764, 1986.

\bibitem{Col_imrn03}
B.~Collins.
\newblock Moments and cumulants of polynomial random variables on unitary
  groups, the {I}tzykson-{Z}uber integral, and free probability.
\newblock {\em Int. Math. Res. Not.}, 2003:953--982, 2003.

\bibitem{forrester06}
P.~J. Forrester.
\newblock Quantum conductance problems and the {J}acobi ensemble.
\newblock {\em J. Phys. A}, 39:6861--6870, 2006.

\bibitem{GouGuaNov_prep11}
I.~Goulden, M.~Guay-Paquet, and J.~Novak.
\newblock Monotone {H}urwitz numbers and the {HCIZ} integral {II}.
\newblock Preprint {\tt arXiv:1107.1001}, 2011.

\bibitem{ha84}
J.~H. Hannay and A.~M. Ozorio~de Almeida.
\newblock Periodic orbits and a correlation function for the semiclassical
  density of states.
\newblock {\em J. Phys. A}, 17:3429--3440, 1984.

\bibitem{heusleretal07}
S.~Heusler, S.~M\"uller, A.~Altland, P.~Braun, and F.~Haake.
\newblock Periodic-orbit theory of level correlations.
\newblock {\em Phys. Rev. Lett.}, 98:044103, 2007.

\bibitem{heusleretal06}
S.~Heusler, S.~M\"uller, P.~Braun, and F.~Haake.
\newblock Semiclassical theory of chaotic conductors.
\newblock {\em Phys. Rev. Lett.}, 96:066804, 2006.

\bibitem{JacVis_Atlas}
D.~M. Jackson and T.~I. Visentin.
\newblock {\em An atlas of the smaller maps in orientable and nonorientable
  surfaces}.
\newblock CRC Press Series on Discrete Mathematics and its Applications.
  Chapman \& Hall/CRC, Boca Raton, FL, 2001.

\bibitem{jw06}
{\relax{Ph}}.~Jacquod and R.~S. Whitney.
\newblock Semiclassical theory of quantum chaotic transport: {P}hase-space
  splitting, coherent backscattering and weak localization.
\newblock {\em Phys. Rev. B}, 73:195115, 2006.

\bibitem{kss09}
B.~A. Khoruzhenko, D.~V. Savin, and H.-J. Sommers.
\newblock Systematic approach to statistics of conductance and shot-noise in
  chaotic cavities.
\newblock {\em Phys. Rev. B}, 80:125301, 2009.

\bibitem{kuipersetal11}
J.~Kuipers, T.~Engl, G.~Berkolaiko, C.~Petitjean, D.~Waltner, and K.~Richter.
\newblock The density of states of chaotic {A}ndreev billiards.
\newblock {\em Phys. Rev. B}, 83:195316, 2011.

\bibitem{kr13}
J.~Kuipers and K.~Richter.
\newblock Transport moments and {A}ndreev billiards with tunnel barriers.
\newblock {\em J. Phys. A}, 46:055101, 2013.

\bibitem{kuipersetal10}
J.~Kuipers, D.~Waltner, C.~Petitjean, G.~Berkolaiko, and K.~Richter.
\newblock Semiclassical gaps in the density of states of chaotic {A}ndreev
  billiards.
\newblock {\em Phys. Rev. Lett.}, 104:027001, 2010.

\bibitem{Landauer57}
R.~Landauer.
\newblock Spatial variation of currents and fields due to localized scatterers
  in metallic conduction.
\newblock {\em IBM J. Res. Dev.}, 1:223--231, 1957.

\bibitem{Landauer88}
R.~Landauer.
\newblock Spatial variation of currents and fields due to localized scatterers
  in metallic conduction.
\newblock {\em IBM J. Res. Dev.}, 33:306--316, 1988.

\bibitem{lv11}
G.~Livan and P.~Vivo.
\newblock Moments of {W}ishart-{L}aguerre and {J}acobi ensembles of random
  matrices: {A}pplication to the quantum transport problem in chaotic cavities.
\newblock {\em Acta Phys. Pol. B}, 42:1081--1104, 2011.

\bibitem{matsumoto12}
S.~Matsumoto.
\newblock General moments of the inverse real {W}ishart distribution and
  orthogonal {W}eingarten functions.
\newblock {\em J. Theor. Prob.}, 25:798--822, 2012.

\bibitem{Mat_prep13}
S.~Matsumoto.
\newblock {W}eingarten calculus for matrix ensembles associated with compact
  symmetric spaces.
\newblock {\em Random Matrices: Theory Appl.}, 2:1350001, 2013.

\bibitem{MatNov_fpsac10}
S.~Matsumoto and J.~Novak.
\newblock Unitary matrix integrals, primitive factorizations, and
  {J}ucys-{M}urphy elements.
\newblock In {\em 22nd International Conference on Formal Power Series and
  Algebraic Combinatorics (FPSAC 2010)}, pages 403--412. DMTCS Proceedings,
  2010.

\bibitem{Mel_jpa90}
P.~A. Mello.
\newblock Averages on the unitary group and applications to the problem of
  disordered conductors.
\newblock {\em J. Phys. A}, 23:4061--4080, 1990.

\bibitem{MelSel_npa80}
P.~A. Mello and T.~H. Seligman.
\newblock On the entropy approach to statistical nuclear reactions.
\newblock {\em Nuc. Phys. A}, 344:489--508, 1980.

\bibitem{melsenetal96}
J.~A. Melsen, P.~W. Brouwer, K.~M. Frahm, and C.~W.~J. Beenakker.
\newblock Induced superconductivity distinguishes chaotic from integrable
  billiards.
\newblock {\em Europhys. Lett.}, 35:7--12, 1996.

\bibitem{melsenetal97}
J.~A. Melsen, P.~W. Brouwer, K.~M. Frahm, and C.~W.~J. Beenakker.
\newblock Superconductor-proximity effect in chaotic and integrable billiards.
\newblock {\em Phys. Scr.}, T69:223--225, 1997.

\bibitem{ms11}
F.~Mezzadri and N.~Simm.
\newblock Moments of the transmission eigenvalues, proper delay times and
  random matrix theory {I}.
\newblock {\em J. Math. Phys.}, 52:103511, 2011.

\bibitem{ms12}
F.~Mezzadri and N.~Simm.
\newblock Moments of the transmission eigenvalues, proper delay times and
  random matrix theory {II}.
\newblock {\em J. Math. Phys.}, 53:053504, 2012.

\bibitem{ms13}
F.~Mezzadri and N.~Simm.
\newblock Tau-function theory of quantum chaotic transport with beta=1,2,4.
\newblock Preprint {\tt arXiv:1206.4584}, 2012.

\bibitem{miller75}
W.~H. Miller.
\newblock The classical {S}-matrix in molecular collisions.
\newblock {\em Adv. Chem. Phys.}, 30:77--136, 1975.

\bibitem{mulleretal09}
S.~M\"uller, S.~Heusler, A.~Altland, P.~Braun, and F.~Haake.
\newblock Periodic-orbit theory of universal level correlations in quantum
  chaos.
\newblock {\em New J. Phys.}, 11:103025, 2009.

\bibitem{mulleretal07}
S.~M\"uller, S.~Heusler, P.~Braun, and F.~Haake.
\newblock Semiclassical approach to chaotic quantum transport.
\newblock {\em New J. Phys.}, 9:12, 2007.

\bibitem{mulleretal04}
S.~M\"uller, S.~Heusler, P.~Braun, F.~Haake, and A.~Altland.
\newblock Semiclassical foundation of universality in quantum chaos.
\newblock {\em Phys. Rev. Lett.}, 93:014103, 2004.

\bibitem{mulleretal05}
S.~M\"uller, S.~Heusler, P.~Braun, F.~Haake, and A.~Altland.
\newblock Periodic-orbit theory of universality in quantum chaos.
\newblock {\em Phys. Rev. E}, 72:046207, 2005.

\bibitem{Note1}
Another name present in the literature is ``Weingarten'' function \cite
  {Wei_jmp78}, even though it was probably Samuel \cite {Sam_jmp80} who first
  defined the function and systematically studied it.

\bibitem{Note2}
Thus, despite the word ``orthogonal'' in the name, it is not the orthogonal
  group $O(N)$. Rather, it can be identified with $U(N)/O(N)$.

\bibitem{Note3}
The traditional definitions are slightly different from ours in ordering the
  second elements of the transpositions.

\bibitem{novaes08}
M.~Novaes.
\newblock Statistics of quantum transport in chaotic cavities with broken
  time-reversal symmetry.
\newblock {\em Phys. Rev. B}, 78:035337, 2008.

\bibitem{novaes12}
M.~Novaes.
\newblock Semiclassical approach to universality in quantum chaotic transport.
\newblock {\em EPL}, 98:20006, 2012.

\bibitem{novaes13}
M.~Novaes.
\newblock Combinatorial problems in the semiclassical approach to quantum
  chaotic transport.
\newblock {\em J. Phys. A}, 46:095101, 2013.

\bibitem{novaes13b}
M.~Novaes.
\newblock A semiclassical matrix model for quantum chaotic transport.
\newblock Preprint, 2013.

\bibitem{ok08}
V.~A. Osipov and E.~Kanzieper.
\newblock Integrable theory of quantum transport in chaotic cavities.
\newblock {\em Phys. Rev. Lett.}, 101:176804, 2008.

\bibitem{ok09}
V.~A. Osipov and E.~Kanzieper.
\newblock Statistics of thermal to shot noise crossover in chaotic cavities.
\newblock {\em J. Phys. A}, 42:475101, 2009.

\bibitem{petitjeanetal09}
C.~Petitjean, D.~Waltner, J.~Kuipers, I.~Adagideli, and K.~Richter.
\newblock Semiclassical approach to the dynamical conductance of a chaotic
  cavity.
\newblock {\em Phys. Rev. B}, 80:115310, 2009.

\bibitem{richter00}
K.~Richter.
\newblock {\em Semiclassical theory of mesoscopic quantum systems}.
\newblock Springer, Berlin, 2000.

\bibitem{rs02}
K.~Richter and M.~Sieber.
\newblock Semiclassical theory of chaotic quantum transport.
\newblock {\em Phys. Rev. Lett.}, 89:206801, 2002.

\bibitem{Sam_jmp80}
S.~Samuel.
\newblock {${\rm U}(N)$} integrals, {$1/N$}, and the {D}e\thinspace
  {W}it-'t\thinspace {H}ooft anomalies.
\newblock {\em J. Math. Phys.}, 21:2695--2703, 1980.

\bibitem{ss06}
D.~V. Savin and H.-J. Sommers.
\newblock Shot noise in chaotic cavities with an arbitrary number of open
  channels.
\newblock {\em Phys. Rev. B}, 73:081307, 2006.

\bibitem{ssw08}
D.~V. Savin, H.-J. Sommers, and W.~Wieczorek.
\newblock Nonlinear statistics of quantum transport in chaotic cavities.
\newblock {\em Phys. Rev. B}, 77:125332, 2008.

\bibitem{Sie02}
M.~Sieber.
\newblock Leading off-diagonal approximation for the spectral form factor for
  uniformly hyperbolic systems.
\newblock {\em J. Phys. A}, 35:L613--L619, 2002.

\bibitem{sr01}
M.~Sieber and K.~Richter.
\newblock Correlations between periodic orbits and their r\^ole in spectral
  statistics.
\newblock {\em Phys. Scr.}, T90:128--133, 2001.

\bibitem{sgb03}
P.~G. Silvestrov, M.~C. Goorden, and C.~W.~J. Beenakker.
\newblock Adiabatic quantization of {A}ndreev quantum billiard levels.
\newblock {\em Phys. Rev. Lett.}, 90:116801, 2003.

\bibitem{sgb03b}
P.~G. Silvestrov, M.~C. Goorden, and C.~W.~J. Beenakker.
\newblock Noiseless scattering states in a chaotic cavity.
\newblock {\em Phys. Rev. B}, 67:241301, 2003.

\bibitem{SmiLerAlt_prb98}
R.~A. Smith, I.~V. Lerner, and B.~L. Altshuler.
\newblock Spectral statistics in disordered metals: A trajectories approach.
\newblock {\em Phys. Rev. B}, 58:10343--10350, Oct 1998.

\bibitem{sss01}
H.-J. Sommers, D.~V. Savin, and V.~V. Sokolov.
\newblock Distribution of proper delay times in quantum chaotic scattering: {A}
  crossover from ideal to weak coupling.
\newblock {\em Phys. Rev. Lett.}, 87:094101, 2001.

\bibitem{Spe_jpa03}
D.~Spehner.
\newblock Spectral form factor of hyperbolic systems: leading off-diagonal
  approximation.
\newblock {\em J. Phys. A}, 36(26):7269--7290, 2003.

\bibitem{TurRic_jpa03}
M.~Turek and K.~Richter.
\newblock Leading off-diagonal contribution to the spectral form factor of
  chaotic quantum systems.
\newblock {\em J. Phys. A-Math. Gen.}, {36}({30}):{L455--L462}, {2003}.

\bibitem{Tutte_GraphTheory}
W.~T. Tutte.
\newblock {\em Graph theory}, volume~21 of {\em Encyclopedia of Mathematics and
  its Applications}.
\newblock Addison-Wesley, Reading, MA, 1984.

\bibitem{vv08}
P.~Vivo and E.~Vivo.
\newblock Transmission eigenvalue densities and moments in chaotic cavities
  from random matrix theory.
\newblock {\em J. Phys. A}, 41:122004, 2008.

\bibitem{wk10}
D.~Waltner and J.~Kuipers.
\newblock Ehrenfest time dependence of quantum transport corrections and
  spectral statistics.
\newblock {\em Phys. Rev. E}, 82:066205, 2010.

\bibitem{waltneretal12}
D.~Waltner, J.~Kuipers, {\relax Ph}.~Jacquod, and K.~Richter.
\newblock Conductance fluctuations in chaotic systems with tunnel barriers.
\newblock {\em Phys. Rev. B}, 85:024302, 2012.

\bibitem{wkr11}
D.~Waltner, J.~Kuipers, and K.~Richter.
\newblock Ehrenfest-time dependence of counting statistics for chaotic
  ballistic systems.
\newblock {\em Phys. Rev. B}, 83:195315, 2011.

\bibitem{Wei_jmp78}
D.~Weingarten.
\newblock Asymptotic behavior of group integrals in the limit of infinite rank.
\newblock {\em J. Math. Phys.}, 19:999--1001, 1978.

\bibitem{whitney07}
R.~S. Whitney.
\newblock Suppression of weak localization and enhancement of noise by
  tunneling in semiclassical chaotic transport.
\newblock {\em Phys. Rev. B}, 75:235404, 2007.

\bibitem{wj06}
R.~S. Whitney and {\relax Ph}.~Jacquod.
\newblock Shot noise in semiclassical chaotic cavities.
\newblock {\em Phys. Rev. Lett.}, 96:206804, 2006.

\bibitem{Zvo_mcm97}
A.~Zvonkin.
\newblock Matrix integrals and map enumeration: an accessible introduction.
\newblock {\em Math. Comput. Modelling}, 26:281--304, 1997.

\end{thebibliography}

\end{document}